\documentclass[final,nomarks]{dmtcs-episciences}

\pdfoutput=1

\usepackage[utf8]{inputenc}
\usepackage{subfigure}
\usepackage{amsmath,amssymb,amsthm,amscd,accents,latexsym,stmaryrd}
\usepackage{amsfonts}
\usepackage{bbm}
\usepackage[all,cmtip]{xy}
\usepackage{enumerate}
\usepackage{url}
\usepackage{tikz}
\usetikzlibrary{cd}
\usetikzlibrary{decorations.text}
\usetikzlibrary{patterns}
\usepackage{mathtools}
\usetikzlibrary{automata,positioning}

\makeatletter
\DeclareFontEncoding{LS2}{}{\@noaccents}
\makeatother
\DeclareFontSubstitution{LS2}{stix}{m}{n}
\DeclareSymbolFont{largesymbolsstix}{LS2}{stixex}{m}{n}
\DeclareMathDelimiter{\lbrbrak}{\mathopen}{largesymbolsstix}{"EE}{largesymbolsstix}{"14}
\DeclareMathDelimiter{\rbrbrak}{\mathclose}{largesymbolsstix}{"EF}{largesymbolsstix}{"15}

\theoremstyle{definition}
\newtheorem{defin}{Definition}[subsection]
\newtheorem{rem}[defin]{Remark}
\newtheorem{rems}[defin]{Remarks}

\newtheorem{ex}[defin]{Example}
\newtheorem{exs}[defin]{Examples}
\theoremstyle{plain}
\newtheorem{theor}[defin]{Theorem}

\newtheorem{prop}[defin]{Proposition}
\newtheorem{cor}[defin]{Corollary}

\def\N{{\mathbb{N}}}
\def\Z{{\mathbb{Z}}}
\def\im{{\mbox{im}}}
\def\length{{\textsc{l}}}
\def\A{{\mathcal{A}}}
\def\B{{\mathcal{B}}}
\def\C{{\mathcal{C}}}
\def\star{{\rm{star}}}

\author{Thomas Kahl}

\title{Weak equivalence of higher-dimensional automata\thanks{This research was partially supported by FCT (\emph{Funda\c c\~ao para a Ci\^encia e a Tecnologia}, Portugal) through project UID/MAT/00013/2013.}}

\affiliation{
  Centro de Matem\'atica, Universidade do Minho, Braga,
  Portugal}

\received{2019-10-31}
\revised{2021-01-14}
\accepted{2021-05-04}

\begin{document}
\publicationdetails{23}{2021}{1}{12}{5884}
\maketitle
\begin{abstract}
  This paper introduces a notion of equivalence for higher-dimensional automata, called weak equivalence. Weak equivalence focuses mainly on a traditional trace language and a new homology language, which captures the overall independence structure of an HDA. It is shown that weak equivalence is compatible with both the tensor product and the coproduct of HDAs and that, under certain conditions, HDAs may be reduced to weakly equivalent smaller ones by merging and collapsing cubes.	
\end{abstract}

\section{Introduction}

\subsection{Higher-dimensional automata} A higher-dimensional automaton (HDA) is an automaton with a supplementary structure consisting of two- and higher-dimensional cubes linking its states and transitions. The underlying automaton of an HDA represents a concurrent system. An \(n\)-cube in an HDA indicates that the \(n\) actions starting at its origin are independent in the sense that they may be executed in any order, or even simultaneously,  without any observable difference. The notion of higher-dimensional automaton goes back to Pratt \cite{Pratt}. The concept used in this paper is essentially a generalization of the one defined by van Glabbeek \cite{vanGlabbeek}. Our definition differs from the one of van Glabbeek in that we consider HDAs over concurrent alphabets and allow labels to be words (see Section \ref{HDAdef}).

\subsection{Weak equivalence} The purpose of this paper is to introduce a concept of equivalence for higher-dimensional automata, called \emph{weak equivalence}. The adjective \emph{weak} is meant to emphasize that the structure two HDAs must have in common to be considered equivalent is reduced to a few essential features. More precisely, two HDAs must satisfy three conditions to be weakly equivalent. The first condition guarantees that two HDAs are \emph{not} weakly equivalent if one of them has unreachable states but the other does not or if one has bad features such as deadlocks but the other does not. 

The second requirement is that weakly equivalent HDAs must have the same \emph{trace language} and the same \emph{fundamental monoid}. These are defined along traditional lines as subsets of the trace monoid associated with the concurrent alphabet of the HDAs (see Section \ref{sectrace}). Higher-dimensional automata with the same trace language behave the same with respect to safety properties that are compatible with the congruence relation induced by the independence relation of the concurrent alphabet. 

The third and last condition for weak equivalence concerns primarily the higher-dimensional structure of HDAs. In \cite{labels}, it has been shown that the cubical homology of an HDA can be equipped with a labeling. In Section \ref{secHL}, this labeling is used to define the \emph{homology language} of an HDA, which reflects its global independence structure. Weakly equivalent HDAs are required to have the same homology language.

\subsection{Weak implementation} 

Weak equivalence is the symmetric closure of a preorder, which we call \emph{weak implementation}. The definition of this preorder is obtained from the one of weak equivalence essentially by replacing equalities by inclusions (see Section \ref{secwe}). Weak implementation is related to morphisms of HDAs in the following way: if \(\A\) and \(\B\) are two HDAs over the same concurrent alphabet and there exists a morphism from \(\A\) to \(\B\) that respects the concurrent alphabet and preserves unreachable and uncoreachable states, then \(\A\) weakly implements \(\B\). This still holds for the more flexible cubical dimaps of HDAs, which have been introduced in \cite{labels}. Cubical dimaps permit one to compare HDAs of different atomicity levels, which is not possible with morphisms: if an HDA is constructed from another one by merging cubes and edge labels, then there will exist a cubical dimap but no morphism between the two HDAs. Cubical dimaps will be discussed in Section \ref{seccubdi}.

\subsection{Parallel composition and nondeterministic sum}

Higher-dimensional automata may be used in different ways to model concurrent systems. For HDAs modeling shared-variable systems, two categorical constructions are particularly important: the tensor product of HDAs, which models the parallel composition of independent concurrent systems, and the coproduct of HDAs, which corresponds to the nondeterministic sum of concurrent systems. We show that the relations of weak equivalence and weak implementation are compatible with the tensor product and, for coaccessible HDAs, with the coproduct (see Section \ref{wedef}).

\subsection{Reduction of HDAs}

In view of the state explosion problem, it is desirable to be able to reduce HDAs to weakly equivalent smaller ones. In \cite{topabs}, conditions have been established under which a so-called \emph{topological abstraction} of an HDA can be constructed by collapsing and merging cubes. Since the relation of topological abstraction is normally stronger than weak equivalence, it is possible to adapt the results of \cite{topabs} to obtain reduction operations that yield weakly equivalent HDAs. This is done in Sections \ref{collapse} and \ref{merging}.

\subsection{Background and related work}

Higher-dimensional automata have been devised by Pratt and van Glabbeek (see \cite{Pratt, vanGlabbeek}). A bibliography on HDAs can be found in \cite{vanGlabbeek}. Descriptions of how HDAs can be used to model concurrent systems are contained in \cite{FGHMR, GaucherProcess, vanGlabbeek, GoubaultMimram, transhda}. 

This paper adopts Winskel and Nielsen's categorical perspective on models for concurrency, according to which the morphisms in a category of objects modeling concurrent systems represent simulations and categorical constructions correspond to composition operators \cite{WinskelNielsen}. From this point of view, weak equivalence is coarser than a kind of simulation equivalence. As we note in Remark \ref{bisimrem}, this does not remain true for history-preserving bisimilarity in the sense of \cite{vanGlabbeek}. A categorical theory of bisimulation, which may be used to define a notion of bisimilarity that is stronger than weak equivalence, is developed in \cite{JoyalNielsenWinskel}. A comparison of different approaches to  simulation is provided in \cite{LynchVaandrager}.  

Weak implementation may be considered a coarse precongruence in the spirit of \cite{vanGlabbeekCoarsest}. The main properties preserved by weak implementation and weak equivalence are the trace language and the homology language. The trace language fits within the framework of 
Mazurkiewicz trace theory. The fundamental material on this subject is contained in \cite{AalbersbergRozenberg, Diekert, DiekertMetivier, DiekertMuscholl, Mazurkiewicz, Mazurkiewicz2}. The definition of the homology language is based on concepts from algebraic topology. Two of the many textbooks in this area are \cite{Dold, Hatcher}. 

The existence of connections between concurrency theory and algebraic topology is at the origin of the field of directed algebraic topology \cite{FGHMR, GrandisBook}. Since the homology language is invariant under cubical dimaps that are homotopy equivalences (see Proposition \ref{HLinc}), it may be considered a directed homotopy invariant of HDAs. However, since it depends on the labeling structure of HDAs rather than on their directed topology, it is not a concept of directed homology as those considered in the literature (see, \textit{e.g.}, \cite{GrandisBook, hgraph}). A brief account of work on directed homology is given in \cite[p. 153]{FGHMR}.

Our results on the reduction of HDAs are inspired by but not directly related to work on partial order reduction \cite{ Godefroid, Peled} and discrete Morse theory \cite{FormanMorseTheory}.

\section{Higher-dimensional automata}

This section presents basic material on precubical sets, concurrent alphabets, and  higher-dimensional automata. The definition of higher-dimensional automata is essentially the one of van Glabbeek \cite{vanGlabbeek}, with the difference that we consider HDAs over concurrent alphabets and allow labels to be words. It is not our intention to provide a comprehensive introduction to the material of this section. For more details, explanations, and examples, the reader is referred to, \textit{e.g.}, \cite{Diekert, Mazurkiewicz, FGHMR, vanGlabbeek}.

\subsection{Precubical sets} \label{precubs}

A \emph{precubical set} is a graded set \(P = (P_n)_{n \geq 0}\) with  \emph{boundary} or \emph{face operators} \(d^k_i\colon P_n \to P_{n-1}\) \(({n>0,}\;k\in\{0,1\},\; i \in \{1, \dots, n\})\) satisfying the relations \(d^k_i\circ d^l_{j}= d^l_{j-1}\circ d^k_i\) \((k,l \in \{0,1\},\; i<j)\). If \(x\in P_n\), we say that \(x\) is of  \emph{degree} or \emph{dimension} \(n\). The elements of degree \(n\) are called the \emph{\(n\)-cubes} of \(P\). The elements of degree \(0\) are also called the \emph{vertices} of \(P\), and the \(1\)-cubes are also called the \emph{edges} of \(P\). A face \(d^k_ix\) is called a \emph{front face} of \(x\) if \(k=0\) and a \emph{back face} of \(x\) if \(k=1\). A \emph{precubical subset} of a precubical set \(P\) is a graded subset of \(P\) that is stable under the boundary operators. A \emph{morphism} of precubical sets is a morphism of graded sets that is compatible with the boundary operators. 

The category of precubical sets can be seen as the presheaf category \({\mathsf{Set}}^{{\square}^{\mbox{\tiny op}}}\) where \(\square\) is the small subcategory of the category of topological spaces whose objects are the standard \(n\)-cubes \([0,1]^n\) \((n \geq 0)\) and whose nonidentity morphisms are composites of the maps \(\delta^k_i\colon [0,1]^n\to [0,1]^{n+1}\) (\(k \in \{0,1\}\), \(n \geq 0\), \(i \in  \{1, \dots, n+1\}\)) given by \(\delta_i^k(u_1,\dots, u_n)= (u_1,\dots, u_{i-1},k,u_i \dots, u_n)\).

\subsection{Tensor product of precubical sets} \label{tensor}

The \emph{tensor product} of two graded sets \(P\) and \(Q\) is the graded set \(P\otimes Q\) given by \[(P\otimes Q)_n = \coprod \limits_{p+q = n} P_p\times Q_q. \] 
If \(P\) and \(Q\) are precubical sets, then \(P\otimes Q\) is a precubical set. For an \(n\)-cube \((x,y) \in P_p\times Q_q\) \((p+q=n)\), the boundary operators are defined by 
\[d_i^k(x,y) = \left\{ \begin{array}{ll} (d_i^kx,y), & 1\leq i \leq p,\\
(x,d_{i-p}^ky), & p < i \leq n.
\end{array}\right.\] 
With respect to this tensor product, the category \({\mathsf{Set}}^{{\square}^{\mbox{\tiny op}}}\) is a monoidal category.

\subsection{Precubical intervals} 

Let \(k\) and \(l\) be two integers such that \(k \leq l\). The \emph{precubical interval}  \(\lbrbrak k,l \rbrbrak\) is the  precubical set defined by \(\lbrbrak k,l \rbrbrak_0 = \{k,\dots , l\}\), \(\lbrbrak k,l \rbrbrak_1 =  \{{[k,k+ 1]}, \dots , {[l- 1,l]}\}\), \(d_1^0[j-1,j] = j-1\), \(d_1^1[j-1,j] = j\), and \(\lbrbrak k,l \rbrbrak_{n} = \emptyset\) for \(n > 1\).

\subsection{Precubical cubes}
	
	The \emph{precubical \(n\mbox{-}\)cube} is the \(n\)-fold tensor product \({\lbrbrak 0,1\rbrbrak^ {\otimes n}}\). Here, we use the convention that \({\lbrbrak 0,1\rbrbrak^ {\otimes 0}}\) is the precubical set \(\lbrbrak 0, 0 \rbrbrak = \{0\}\). The only element of degree \(n\) in \(\lbrbrak 0,1\rbrbrak^ {\otimes n}\) will be denoted by \(\iota_n\). We thus have \(\iota_0 = 0\) and \(\iota_n = (\underbracket[0.5pt]{ [0,1] ,\dots , [ 0,1]}_{n\; {\text{times}}})\) for \(n>0\). Given an element \(x\) of degree \(n\) of a precubical set \(P\), there exists a unique morphism of precubical sets \(\lbrbrak 0,1\rbrbrak ^{\otimes n}\to P\) that sends \(\iota_n\) to \(x\). This morphism will be denoted by \(x_{\sharp}\). We say that \(x\) is \emph{regular} if \(x_{\sharp}\) is injective and that \(x\) is \emph{weakly regular} if the restrictions of \(x_{\sharp}\) to the graded subsets \((\lbrbrak 0,1 \rbrbrak\setminus\{1\})^ {\otimes n}\) and \((\lbrbrak 0,1 \rbrbrak\setminus \{0\})^ {\otimes n}\) of \(\lbrbrak 0,1 \rbrbrak^ {\otimes n}\) are injective. If all elements of \(P\) are (weakly) regular, we say that \(P\) is \emph{(weakly) regular}.

\subsection{Paths} \label{paths}
A \emph{path} of \emph{length \(l\)} \((l \geq 0)\) from a vertex \(v\) of a precubical set \(P\) to a vertex \(w\) is a morphism of precubical sets \(\omega \colon \lbrbrak 0,l \rbrbrak \to P\) such that \(\omega(0) = v\) and \(\omega(l) = w\). If \(\omega\) is a path of length \(l\), we write \(\length_\omega = l\). The set of paths in \(P\) is denoted by \(P^{\mathbb I}\). The \emph{concatenation} \(\omega \cdot\nu\)  of two paths \(\omega \colon \lbrbrak 0,k \rbrbrak \to P\) and \(\nu \colon \lbrbrak 0,l \rbrbrak \to P\) with \(\omega (k) = \nu (0)\) is defined in the obvious way. Note that every path in \(P\) of positive length can be uniquely written as a finite concatenation of paths of the form \(x_{\sharp}\) where \(x \in P_{1}\).

\subsection{Dihomotopy} Two paths \(\omega\) and \(\nu \) in a precubical set \(P\) are said to be \emph{elementarily dihomotopic} if there exist paths \(\alpha, \beta \in P^ {\mathbb I}\) and an element \(z \in P_2\) such that \(d_1^0d_1^0z = \alpha (\length_\alpha)\), \(d_1^1d_1^1z = \beta (0)\) and \[\{\omega ,\nu \} = \{{\alpha \cdot (d_1^0z)_{\sharp} \cdot (d_2^ 1z)_{\sharp} \cdot \beta}, {\alpha \cdot (d_2^0z)_{\sharp} \cdot (d_1^ 1z)_{\sharp} \cdot \beta} \}.\]  The \emph{dihomotopy} relation, denoted by \(\sim\), is the equivalence relation generated by elementary dihomotopy \cite{FGHMR} (see Figure \ref{dihompaths} for a picture).

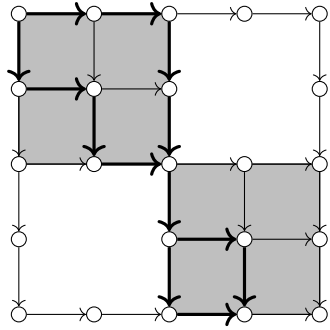
\begin{figure} 
	\center
	\begin{tikzpicture}[initial text={},on grid]

	\path[draw, fill=lightgray] (0,2)--(2,2)--(2,4)--(0,4)--cycle
	(2,0)--(4,0)--(4,2)--(2,2)--cycle;

	\node[state,minimum size=0pt,inner sep =2pt,fill=white] (q_0)   {}; 
	
	\node[state,minimum size=0pt,inner sep =2pt,fill=white] (q_1) [right=of q_0,xshift=0cm] {};
	
	\node[state,minimum size=0pt,inner sep =2pt,fill=white] (q_2) [right=of q_1,xshift=0cm] {};   
	
	\node[state,minimum size=0pt,inner sep =2pt,fill=white] [above=of q_0, yshift=0cm] (q_3)   {};

	\node[state,minimum size=0pt,inner sep =2pt,fill=white] (q_5) [above=of q_2,xshift=0cm] {}; 
	
	\node[state,minimum size=0pt,inner sep =2pt,fill=white] (q_6) [above=of q_3,yshift=0cm] {};
	
	\node[state,minimum size=0pt,inner sep =2pt,fill=white] (q_7) [right=of q_6,yshift=0cm] {};

	\node[state,minimum size=0pt,inner sep =2pt,fill=white] (q_4) [above=of q_7,yshift=0cm] {};

	\node[state,minimum size=0pt,inner sep =2pt,fill=white] (q_8) [above=of q_5,yshift=0cm] {};
	
	\node[state,minimum size=0pt,inner sep =2pt,fill=white] (q_9) [above=of q_6,yshift=0cm] {};
	
	\node[state,minimum size=0pt,inner sep =2pt,fill=white] (q_10) [above=of q_8,yshift=0cm] {};
	
	\node[state,minimum size=0pt,inner sep =2pt,fill=white] (q_11) [above=of q_9,yshift=0cm] {};
	
	\node[state,minimum size=0pt,inner sep =2pt,fill=white] (q_12) [right=of q_11,xshift=0cm] {};
	
	\node[state,minimum size=0pt,inner sep =2pt,fill=white] (q_13) [above=of q_10,yshift=0cm] {};

	\node[state,minimum size=0pt,inner sep =2pt,fill=white] (p_0) [right=of q_2,xshift=0cm]  {}; 
	
	\node[state,minimum size=0pt,inner sep =2pt,fill=white] (p_1) [right=of p_0,xshift=0cm] {};

	\node[state,minimum size=0pt,inner sep =2pt,fill=white] (p_3) [above=of p_0,xshift=0cm] {};

	\node[state,minimum size=0pt,inner sep =2pt,fill=white] (p_4) [above=of p_1,xshift=0cm] {};

	\node[state,minimum size=0pt,inner sep =2pt,fill=white] (p_6) [above=of p_0,yshift=1cm] {};
	
	\node[state,minimum size=0pt,inner sep =2pt,fill=white] (p_7) [above=of p_4,yshift=0cm] {};
	
	\node[state,minimum size=0pt,inner sep =2pt,fill=white] [above=of p_7, yshift=0cm] (p_8)   {};

	\node[state,minimum size=0pt,inner sep =2pt,fill=white] (p_11) [right=of q_13,yshift=0cm] {};
	
	\node[state,minimum size=0pt,inner sep =2pt,fill=white] (p_12) [right=of p_11,xshift=0cm] {};

	\path[->] 
	(q_0) edge[above]  node {} (q_1)
	(q_1) edge[above]  node {} (q_2)
	(q_2) edge[very thick,above]  node {} (p_0)    
	(p_0) edge[above]  node {} (p_1)
	
	(q_3) edge[left]  node {} (q_0)
	(q_5) edge[very thick,right]  node {} (q_2)
	(p_3) edge[very thick,right]  node {} (p_0)
	(p_4) edge[left]  node {} (p_1)
	
	(q_5) edge[very thick,above]  node {} (p_3)
	(p_3) edge[above]  node {} (p_4)
	
	(q_6) edge[left]  node {} (q_3)
	(q_8) edge[very thick,left]  node {} (q_5)    
	(p_6) edge[left]  node {} (p_3)	
	(p_7) edge[left]  node {} (p_4)
	
	(q_6) edge[above]  node {} (q_7)
	(q_7) edge[very thick,above]  node {} (q_8)
	(q_8) edge[above]  node {} (p_6)
	(p_6) edge[above]  node {} (p_7)
	
	(q_9) edge[left]  node {} (q_6)
	(q_4) edge[very thick,left]  node {} (q_7)    
	(q_10) edge[very thick,right]  node {} (q_8)
	(p_8) edge[left]  node {} (p_7)

	(q_9) edge[very thick,above]  node {} (q_4)
	(q_4) edge[above]  node {} (q_10)
	
	(q_11) edge[very thick,left]  node {} (q_9)
	(q_12) edge[left]  node {} (q_4)   
	(q_13) edge[very thick,right]  node {} (q_10)
	(p_12) edge[left]  node {} (p_8)
	
	(q_11) edge[very thick,above]  node {} (q_12)
	(q_12) edge[very thick,above]  node {} (q_13)
	(q_13) edge[above]  node {} (p_11)
	(p_11) edge[above]  node {} (p_12)
	;
	\end{tikzpicture}
	\caption{Dihomotopic paths}\label{dihompaths} 	
\end{figure}

\subsection{Free monoids} Let \(\Sigma\) be an alphabet, \textit{i.e.}, a set. The free monoid over \(\Sigma\) will be denoted by \(\Sigma^*\). The unit element of \(\Sigma^*\), which is the empty word, will be denoted by \(1\). Given a string \(m \in \Sigma^*\), we will write \(|m|\) to denote its \emph{length}, \textit{i.e.}, the unique integer \(n\) such that \(m \in \Sigma ^n\). We say that a string \(m\in \Sigma^*\) \emph{contains} an element \(a \in \Sigma\) if \(m \notin (\Sigma \setminus \{a\})^*\).

\subsection{Concurrent alphabets} A \emph{concurrent alphabet} is a pair \((\Sigma, D)\) where \(\Sigma\) is an alphabet and \(D\) is a reflexive and symmetric relation on \(\Sigma\) (see, \textit{e.g.}, \cite{Diekert, Mazurkiewicz2}). The relation \(D\) is called the \emph{dependence relation} of the concurrent alphabet, and the complement of \(D\) is called the associated \emph{independence relation}. A \emph{morphism} of concurrent alphabets \((\Sigma, D) \to (\Sigma', D')\) is a map \(\sigma \colon \Sigma \to \Sigma'\) such that \({(\sigma\times \sigma)^{-1}(D')} \subseteq D\).

The category of concurrent alphabets is a symmetric monoidal category with respect to the \emph{tensor product} defined by \({(\Sigma_1, D_1) \otimes (\Sigma_2,D_2)} = (\Sigma_1 \amalg \Sigma_2, D_\otimes)\) where \(D_\otimes\) is the union of the images of the canonical maps \[D_i \hookrightarrow \Sigma_i \times \Sigma_i \to (\Sigma_1 \amalg \Sigma_2) \times (\Sigma_1 \amalg \Sigma_2).\]
The tensor product is different from the \emph{coproduct} \({(\Sigma_1, D_1) \amalg (\Sigma_2,D_2)}\), which is the concurrent alphabet \((\Sigma_1 \amalg \Sigma_2, D_\amalg)\) where \(D_\amalg\) is the union of \(D_\otimes\) and the subsets \(\Sigma_1 \times \Sigma_2\) and  \(\Sigma_2 \times \Sigma_1\) of \((\Sigma_1 \amalg \Sigma_2) \times (\Sigma_1 \amalg \Sigma_2)\). While an element of \(\Sigma_1\) and an element of \(\Sigma_2\) are independent in the tensor product, they are dependent in the coproduct.

\subsection{Trace monoids}

Let \((\Sigma , D)\) be a concurrent alphabet, and let \(\equiv\) denote the congruence relation induced by the associated independence relation, \textit{i.e.}, the smallest congruence relation in \(\Sigma^*\) such that \(ab \equiv ba\) for all \({(a,b) \in (\Sigma \times \Sigma) \setminus D}\). The quotient monoid \(\Sigma ^*/\equiv\) is called the \emph{trace monoid} of \((\Sigma,D)\) and is denoted by \(M(\Sigma,D)\) (see, \textit{e.g.}, \cite{Diekert, Mazurkiewicz2}).  Congruent elements of $\Sigma^*$ have the same length. The \emph{length} of an element $m \in M(\Sigma,D)$ may thus be defined by $|m| = |x|$ where $x \in m$. A morphism of concurrent alphabets \(\sigma \colon (\Sigma, D) \to (\Sigma', D')\) induces a monoid homomorphism \(M(\sigma) \colon M(\Sigma, D) \to M(\Sigma', D')\).

Given two concurrent alphabets \((\Sigma_1,D_1)\) and \((\Sigma_2,D_2)\), the homomorphisms induced by the canonical inclusions \((\Sigma_i,D_i) \to (\Sigma_1, D_1) \otimes (\Sigma_2, D_2)\) embed \(M(\Sigma_1,D_1)\) and \(M(\Sigma_2,D_2)\) as submonoids in \(M((\Sigma_1,D_1) \otimes (\Sigma_2,D_2))\). The multiplication \({(m_1, m_2) \mapsto m_1m_2}\) defines a natural isomorphism of monoids \[M(\Sigma_1,D_1) \times M(\Sigma_2,D_2) \to M((\Sigma_1,D_1) \otimes (\Sigma_2,D_2)).\]
The homomorphisms induced by the inclusions \((\Sigma_i,D_i) \to {(\Sigma_1, D_1) \amalg (\Sigma_2, D_2)}\) induce an isomorphism from the free product  \(M(\Sigma_1,D_1) \ast M(\Sigma_2,D_2)\) (which is the coproduct in the category of monoids) to \(M((\Sigma_1,D_1) \amalg (\Sigma_2,D_2))\). 

\subsection{Higher-dimensional automata} \label{HDAdef}
A \emph{higher-dimensional automaton (over a concurrent alphabet)}  is a tuple \[\A = (P_{\A},I_{\A},F_{\A}, \Sigma_\A, D_\A, \lambda_{\A})\] where \(P_{\A}\) is a precubical set, \(I_{\A} \in (P_{\A})_0\) is an \emph{initial state}, \(F_{\A} \subseteq (P_{\A})_0\) is a (possibly empty) set of \emph{final states}, \((\Sigma_\A, D_\A)\) is a concurrent alphabet, and \(\lambda_\A \colon (P_{\A})_1 \to \Sigma_\A^*\) is a \emph{labeling function}. These data are subject to the following two conditions:
\begin{enumerate}[\;\;(1)]
	\item For all \(x \in (P_{\A})_2\) and \(i \in \{1,2\}\), \(\lambda_{\A} (d_i^0x) = \lambda_{\A} (d_i^1x)\). 	
	\item For all \(x \in (P_{\A})_2\) and \((a, b) \in D_\A\), \(\lambda_\A(d^0_1x)\) does not contain \(a\) or \(\lambda_\A(d^0_2x)\) does not contain \(b\). 
\end{enumerate}
We say that an HDA \(\B\) is a \emph{sub-HDA} of an HDA \(\A\) and write \(\B \subseteq \A\) if \(P_\B\) is a precubical subset of \(P_\A\), \(I_\B = I_\A\), \(F_\B = F_\A\cap (P_\B)_0\), \(\Sigma_\B \subseteq \Sigma _\A\), \(D_\B = D_\A\cap(\Sigma _\B \times \Sigma_\B)\), and \(\lambda_\B = \lambda_\A|_{(P_\B)_1}\). An HDA \(\A\) is said to be \emph{(weakly) regular} if the precubical set \(P_\A\) is (weakly) regular. A \emph{morphism} from an HDA \(\A\) to an HDA \(\B\) is a pair \((f,\sigma)\) consisting of a morphism of precubical sets \(f\colon P_\A \to P_\B\) and a morphism of concurrent alphabets \(\sigma \colon (\Sigma_\A, D_\A) \to (\Sigma_\B, D_\B)\) such that \(f(I_{\A}) = I_{\B}\), \(f(F_{\A}) \subseteq F_{\B}\), and \(\lambda_{\B}(f(x)) = \sigma^*(\lambda_{\A}(x))\) for all \(x \in (P_\A)_1\).

Our definition of higher-dimensional automata differs in two points from the one of van Glabbeek \cite{vanGlabbeek}: first, we consider HDAs over concurrent alphabets, and second, we allow labels to be words. Condition (2) above is introduced in the definition of HDAs to guarantee that the independence relation represented by the cubes of an HDA is compatible with the one associated with the concurrent alphabet. An HDA in the sense of \cite{vanGlabbeek} can be seen as an HDA in our sense, at least if it does not admit squares where all edges have the same label. Indeed, given such an HDA, one can define a canonical dependence relation on the alphabet by declaring two actions dependent if there is no square having both of them on its boundary. Condition (2) is then automatically satisfied. We allow labels to be words in order to be able to declare sequences of actions atomic. Another possibility opened up by this modification of van Glabbeek's definition of HDAs is to use the unit of the free monoid on the alphabet to label invisible actions.

\subsection{Labels of paths} 

Let \(\A\) be an HDA. The \emph{extended labeling function} of \(\A\) is the map \(\overline{\lambda}_\A \colon P_\A^{\mathbb I} \to \Sigma_\A^*\) defined as follows: If \(\omega = {x_{1\sharp} \cdot \cdots \cdot x_{k\sharp}}\) for a sequence \((x_1, \dots , x_k)\) of elements of \((P_\A)_1\) such that \(d_1^0x_{j+1} = d_1^1x_j\) \({(1\leq j < k)}\), then we set \(\overline{\lambda}_\A (\omega) = \lambda_\A (x_1) \cdot \cdots \cdot \lambda_\A (x_k)\); if \(\omega\) is a path of length \(0\), then we set  \(\overline{\lambda}_\A (\omega ) = 1\). By conditions (1) and (2) in the definition of HDAs, dihomotopic paths have congruent labels.

\subsection{Tensor product of HDAs} The \emph{tensor product} of two HDAs \(\A\)  and \(\B\) is the HDA \(\A\otimes \B\) defined by \(P_{\A\otimes \B}  = P_{\A}\otimes P_{\B}\), \(I_{\A \otimes \B} = (I_{\A},I_{\B})\), \(F_{\A \otimes \B} = F_{\A}\times F_{\B}\), \((\Sigma_{\A \otimes \B}, D_{\A \otimes \B}) =  (\Sigma_{\A}, D_{\A}) \otimes (\Sigma_{\B},D_{\B})\), and
\[\lambda_{\A\otimes \B} (x,y) = \left \{\begin{array}{ll}
\lambda_{\A}(x), & (x,y) \in (P_{\A})_1\times (P_{\B})_0,\\
\lambda_{\B}(y), & (x,y) \in (P_{\A})_0\times (P_{\B})_1.
\end{array}\right.\]
With respect to this tensor product, the category of HDAs is a (nonsymmetric) monoidal category. The tensor product of HDAs models the parallel composition of independent concurrent systems (for a detailed discussion, see \cite{transhda}).

\subsection{Coproduct of HDAs} 
\begin{sloppypar}
The \emph{coproduct} of two HDAs \(\A\)  and \(\B\) is the HDA \(\A + \B\) where \(P_{\A+ \B}\) is the precubical subset \({P_\A \otimes \{I_\B\} \cup \{I_\A\}\otimes P_\B}\) of \(P_{\A}\otimes P_{\B}\), \(I_{\A + \B} = (I_{\A},I_{\B})\), \(F_{\A + \B} = F_{\A}\times \{I_\B\} \cup \{I_\A\}\times  F_{\B}\), \((\Sigma_{\A + \B}, D_{\A + \B}) =  (\Sigma_{\A}, D_{\A}) \amalg (\Sigma_{\B},D_{\B})\), and
\[\lambda_{\A + \B} (x,y) = \left \{\begin{array}{ll}
\lambda_{\A}(x), & (x,y) \in (P_{\A})_1\times \{I_\B\},\\
\lambda_{\B}(y), & (x,y) \in \{I_\A\} \times (P_{\B})_1.
\end{array}\right.\]
The coproduct of HDAs models the nondeterministic sum of concurrent systems, \textit{i.e.}, the combined system where initially one of the constituent systems is chosen nondeterministically (for more details, see \cite{transhda}). 
\end{sloppypar}

\section{Cubical dimaps} \label{seccubdi}

In their categorical approach to models of concurrency, Winskel and Nielsen \cite{WinskelNielsen} emphasize the importance of the morphisms in a category of objects modeling concurrent systems as a means to express relationships between systems. Unfortunately, the morphisms of HDAs defined in the previous section are too rigid for this purpose in at least two respects. First, they do not permit one to relate HDAs of different atomicity levels. For example, although the HDAs \({\circ \xrightarrow{a} \circ \xrightarrow{b} \circ}\) and \({\circ \xrightarrow{ab} \circ}\) clearly model the same system, there does not exist any morphism between them. Second, there is normally no morphism and, in particular, no isomorphism between the tensor products \(\A\otimes \B\) and \(\B \otimes \A\). This is inconsistent with the fact that the tensor product of HDAs models the parallel composition of independent systems, which is a symmetric operation. In order to address these problems, cubical dimaps (directed maps) have been introduced in \cite{labels}. Roughly speaking, a cubical dimap between two HDAs is a continuous map between their geometric realizations that sends cubes in an order-preserving way to subdivided cubes and that preserves labels of paths. There exists a cubical dimap from \({\circ \xrightarrow{ab} \circ}\) to \({\circ \xrightarrow{a} \circ \xrightarrow{b} \circ}\), and the category of HDAs and cubical dimaps is a symmetric monoidal category. In this section, we collect the main facts about cubical dimaps. More details can be found in \cite{labels}. All topological spaces considered are compactly generated Hausdorff spaces, and constructions such as products are performed in the category of these spaces (see \cite{SteenrodConvenient}).

\subsection{Geometric realization} \label{geomreal}
\begin{sloppypar}
The \emph{geometric realization} of a precubical set \(P\) is the quotient space \[|P|=\left(\coprod _{n \geq 0} P_n \times [0,1]^n\right)/\sim\] where the sets \(P_n\) are given the discrete topology and the equivalence relation is generated by
\[(d^k_ix,u) \sim (x,\delta_i^k(u)), \quad  x \in P_{n+1},\; u\in [0,1]^n,\; i \in  \{1, \dots, n+1\},\; k \in \{0,1\}.\] 
The geometric realization of a morphism of precubical sets \({f\colon P \to Q}\) is the continuous map \({|f|\colon |P| \to |Q|}\) given by \(|f|([x,u])= [f(x),u]\).
\end{sloppypar}
The geometric realization of a precubical set \(P\) is a CW complex. The \(n\)-skeleton of \(|P|\) is the geometric realization of the precubical subset \(P_{\leq n}\) of \(P\) defined by \((P_{\leq n})_i = P_i\) for \(i \leq n\) and \((P_{\leq n})_i = \emptyset\) for \(i > n\). The geometric realization of the precubical interval \(\lbrbrak k,l \rbrbrak\) may be identified with the closed interval \([k,l]\) by means of the homeomorphism \(|\lbrbrak k,l \rbrbrak| \to [k,l]\) given by \([j,()] \mapsto j\) and \([[j-1,j],t] \mapsto j-1+t\). The natural homeomorphism  \({|P\otimes Q| \to |P| \times |Q|}\) given by
\[[(x,y),u] \mapsto  ([x,(u_1,\dots , u_{p})],[y,(u_{p+1}, \dots, u_{p+q})]),\quad (x,y) \in P_p\times Q_q,\,u \in [0,1]^{p+q}\]
permits us to identify the spaces \(|P\otimes Q|\) and \(|P| \times |Q|\).

\subsection{Cubical dimaps of precubical sets} \label{cubdi}
\begin{sloppypar}
An \emph{elementary cubical dimap} from a precubical set \(P\) to a precubical set \(Q\) is a continuous map \({f\colon |P| \to |Q|}\) such that the following two conditions hold:
\end{sloppypar}
\begin{enumerate}[\;\;(1)]
	\item For every vertex \(x\in P_0\), there exists a (necessarily unique) vertex \(y\in Q_0\) such that \(f([x,()]) = [y,()]\).
	\item For every element \(x\in P_n\) \((n>0)\), there exist integers \(l_1, \dots, l_n\geq 1\), a morphism of precubical sets \(\chi \colon \bigotimes_{i=1}^n \lbrbrak 0,l_i\rbrbrak \to Q\), a permutation \(\theta \in S_n\), and increasing homeomorphisms \({\phi_i \colon [0,1] \to  [0,l_i ]}\) \((i \in \{1, \dots, n\})\) such that the following diagram, in which \(t_{\theta}\) is given by \(t_{\theta}(x_1, \dots, x_n) = (x_{\theta(1)}, \dots , x_{\theta(n)})\), is commutative: 	
	\[\xymatrix{
		|\lbrbrak	0, 1\rbrbrak^{\otimes n}| \ar[r]^(0.55){\approx} \ar[d]_{|x_{\sharp}|} & [0,1]^n \ar[r]^{t_{\theta}} & [0,1]^n \ar[rr]^{\phi_1 \times \cdots \times \phi_n} && \prod \limits_{i=1}^n [0,l_i]  \ar[r]^(0.45){\approx} & |\bigotimes \limits_{i=1}^n \lbrbrak 0,l_i\rbrbrak | \ar[d]^{|\chi|}\\
		|P| \ar[rrrrr]_{f} &&&&& |Q|
	}\]

\end{enumerate}
By \cite[Prop. 6.2.4]{labels}, the objects in condition (2) are uniquely determined by \(f\) and \(x\). A \emph{cubical dimap} of precubical sets is a finite composite of elementary cubical dimaps. It can be shown that not all cubical dimaps are elementary. For example, there exists a nonelementary cubical dimap from \(\lbrbrak 0,1 \rbrbrak^{\otimes 2}\) to the precubical set composed of two squares \(a\) and \(b\) such that \(d^1_1a = d^0_2b\). By construction, a cubical dimap is a cellular map. It is clear that condition (1) above holds for arbitrary cubical dimaps and not only for elementary ones. Therefore a cubical dimap \(f\colon |P| \to |Q|\) induces a \emph{vertex map} \(f_0\colon P_0 \to Q_0\), which sends a vertex \(x \in P_0\) to the unique vertex \(y \in Q_0\) such that \(f([x,()]) = [y,()]\). 

The geometric realization of a morphism of precubical sets is an elementary cubical dimap. Hence the presheaf category of precubical sets can be seen as a wide subcategory of the category of precubical sets and cubical dimaps.  Another important class of cubical dimaps is given by subdivisions in the sense of \cite{weakmor}: A subdivision of a precubical set \(P\) consists of a precubical set \(Q\) and a homeomorphism \(|P| \to |Q|\) that is an elementary cubical dimap such that the permutation in condition (2) of the definition is always the identity. In this situation we may, of course, also view \(P\) as obtained from \(Q\) by merging cubes.

The category of precubical sets and cubical dimaps is a symmetric monoidal category with respect to the usual tensor product of precubical sets. The tensor product of two cubical dimaps \(f\colon |P| \to |P'|\) and \(g \colon |Q| \to |Q'|\) is the composite
\[f\otimes g\colon |P\otimes Q| \xrightarrow{\approx} |P|\times |Q| \xrightarrow{f\times g}  |P'| \times |Q'|\xrightarrow{\approx} |P'\otimes Q'|, \]
which is indeed a cubical dimap. The braiding of the symmetric monoidal structure is the homeomorphism 
\[|P\otimes Q| \xrightarrow{\approx} |P| \times |Q| \xrightarrow{\approx} |Q| \times |P| \xrightarrow{\approx} |Q\otimes P|,\]	
which is an elementary cubical dimap.

\subsection{Cubical dimaps and paths} 

Let \(f\colon |P| \to |Q|\) be a cubical dimap of precubical sets, and let \(\omega \colon \lbrbrak 0,k\rbrbrak \to P\) be a path. By \cite[Prop. 6.5.1]{labels}, there exist a unique integer \(l\), a unique  path \(f^{\mathbb I}(\omega) \colon \lbrbrak 0, l\rbrbrak \to Q\), and a unique increasing homeomorphism \(\phi \colon  [0,k] \to [0,l]\) such that the following diagram commutes:	
\[\xymatrix{
	|\lbrbrak	0, k\rbrbrak| \ar[r]^(0.55){\approx} \ar[d]_{|\omega|} & [0,k] \ar[r]^{\phi} & [0,l] \ar[r]^(0.45){\approx}  &  |\lbrbrak 0,l\rbrbrak | \ar[d]^{|f^{\mathbb I}(\omega)|}\\
	|P| \ar[rrr]_{f} &&& |Q|
}\]
We remark that \(\length_{f^{\mathbb I}(\omega)} = 0\) if \(\length_\omega = 0\) and that \(f^{\mathbb I}(\omega)\) is a path from \(f_0(\omega(0))\) to \(f_0(\omega(\length_\omega))\). Note also that if \(f\) is the geometric realization of a morphism of precubical sets \(h\colon P \to Q\), then \(f^{\mathbb I}(\omega) = h\circ \omega\). By adapting  the arguments given in \cite{weakmor} in the context of weak morphisms, it is easily seen that the construction of \(f^{\mathbb I}(\omega)\) is compatible with composition of cubical dimaps, concatenation of paths, and dihomotopy.

\subsection{Cubical dimaps of HDAs}

An \emph{elementary cubical dimap} from an HDA \(\A\) to an HDA \(\B\) is a pair \((f,\sigma)\) consisting of an elementary cubical dimap \(f\colon |P_\A| \to |P_\B|\) and a morphism of concurrent alphabets \(\sigma \colon (\Sigma_\A, D_\A) \to (\Sigma_\B, D_\B)\) such that \(f_0(I_\A) = I_\B\), \(f_0(F_\A) \subseteq F_\B\), and \(\overline{\lambda}_\B\circ f^{\mathbb I} = \sigma^* \circ \overline{\lambda}_\A\). If \((g,\sigma) \colon \A \to \B\) is a morphism of HDAs, then its geometric realization \((|g|,\sigma)\) is an elementary cubical dimap from \(\A\) to \(\B\). A \emph{cubical dimap} of HDAs is a finite componentwise composite of elementary cubical dimaps. Note that if \((f,\sigma)\) is a cubical dimap of HDAs, then \(f\) is a cubical dimap of precubical sets and the above conditions for elementary cubical dimaps hold. 

The coproduct of HDAs is the coproduct in the category of HDAs and cubical dimaps. The tensor product of HDAs turns this category into a  symmetric monoidal category. The tensor product of cubical dimaps and the natural isomorphisms of the symmetric monoidal structure are defined componentwise.

\section{The trace language of an HDA} \label{sectrace}

The trace language of an HDA, which is defined in this section, describes its possible finite behavior. It contains the information necessary to decide whether an HDA satisfies a given saturated safety property, \textit{i.e.}, a safety property that is compatible with the congruence relation of the concurrent alphabet of the HDA. In addition to the trace language of an HDA, we define its fundamental monoid. We show that the trace language and the fundamental monoid behave well with respect to cubical dimaps and establish formulas to compute them for tensor products and coproducts. The trace language of an HDA is a trace language in the sense of Mazurkiewicz trace theory. References on this subject are \cite{AalbersbergRozenberg, Diekert, DiekertMetivier, DiekertMuscholl, Mazurkiewicz, Mazurkiewicz2}.

\subsection{Saturated safety properties} 
 
	Let \(\Sigma\) be an alphabet. Following van Glabbeek \cite{vanGlabbeekCoarsest}, we say that a \emph{safety property} is given by a set \(B \subseteq \Sigma^*\). An HDA $\A$ with \(\Sigma_\A = \Sigma\) \emph{satisfies} this safety property, \(\A \models \mathit{safety}(B)\), when for every  path \(\omega \in P_\A^{\mathbb I}\) with \(\omega (0) = I_\A\), \(\overline{\lambda}_\A(\omega)\notin B\Sigma^*\). Note that since we allow labels to be words, it is not enough to require \(\overline{\lambda}_\A(\omega)\notin B\). Note also that \(B\) and \(B\Sigma^*\) define the same safety property: for any HDA \(\A\) with \(\Sigma_\A = \Sigma\), \(\A \models \mathit{safety}(B) \Leftrightarrow \A \models \mathit{safety}(B\Sigma^*)\). If \((\Sigma, D)\) is a concurrent alphabet, a safety property given by a subset \(B \subseteq \Sigma^*\) is called \emph{saturated} if for any two congruent elements \(m, m' \in \Sigma^*\), \(m \in B\Sigma^* \Leftrightarrow m' \in B\Sigma^*\).

\subsection{Prefixes} Let \(M\) be a monoid. We say that an element \(v \in M\) is a \emph{prefix} of an element \(u \in M\) and write \(v \preceq u\) if there exists an element \(w \in M\) such that \(u = vw\). The relation \(\preceq\) is a preorder on  \(M\). If \(M\) is free or a trace monoid, then \(\preceq\) is a partial order.

\subsection{Trace language} The \emph{trace language} of an HDA \(\A\) is the set 
\[TL(\A) = \{v \in M(\Sigma_\A,D_\A)\, |\, \exists \,\omega \in P_\A^{\mathbb I} : \omega(0) = I_\A,\,  v \preceq [\overline{\lambda}_\A(\omega)]\}.\]
Note that although we use the same notation, \(TL(\A)\) is different from the trace language defined in \cite{weakmor}. By the next two propositions, the trace language contains exactly the information needed to determine which saturated safety properties are satisfied by an HDA.

\begin{prop} \label{satsafe}
	Let \((\Sigma, D)\) be a concurrent alphabet, and let \(B \subseteq \Sigma^*\) define a saturated safety property. Then for any HDA \(\A\) such that \((\Sigma_\A,D_\A) = (\Sigma,D)\),  
	\(\A \models \mathit{safety}(B)\) if and only if \(TL(\A) \subseteq \{[s]\;|\; s \in \Sigma^* \setminus (B\Sigma^*)\}\). 	
\end{prop}

\begin{proof}
	Suppose first that \(\A \models \mathit{safety}(B)\). Consider \(v = [x] \in TL(\A)\). Let \(\omega\) be a path in \(\A\) such that \(\omega(0) = I_\A\) and \(v \preceq [\overline{\lambda}_\A(\omega)]\), and let \(w = [y] \in M(\Sigma, D)\) such that \([\overline{\lambda}_\A(\omega)] = vw = [xy]\). Since \(\A \models \mathit{safety}(B)\), \(\overline{\lambda}_\A(\omega) \notin B\Sigma^*\). Since \(xy \equiv \overline{\lambda}_\A(\omega)\), also \(xy \notin B\Sigma^*\). Thus \(x \notin B\Sigma^*\). Hence \(v = [x] \in \{[s]\;|\; s \in \Sigma^* \setminus (B\Sigma^*)\}\).

	Suppose now that \(TL(\A) \subseteq \{[s]\;|\; s \in \Sigma^* \setminus (B\Sigma^*)\}\). Let \(\omega\) be a path in \(\A\) such that \(\omega(0) = I_\A\). Consider a prefix \(x \preceq \overline{\lambda}_\A(\omega)\). Then \([x] \preceq [\overline{\lambda}_\A(\omega)]\) and therefore \([x] \in TL(\A)\subseteq \{[s]\;|\; s \in \Sigma^* \setminus (B\Sigma^*)\}\). Hence \(x \equiv y\) for some \(y \in \Sigma^* \setminus (B\Sigma^*)\). This implies \(x \in \Sigma^* \setminus (B\Sigma^*)\). Since \(B \subseteq B\Sigma^*\), \(x \notin B\). It follows that \(\overline{\lambda}_\A(\omega) \notin B\Sigma^*\). 	
\end{proof}

\begin{prop} \label{TLsatsafe}
	Two HDAs \(\A\) and \(\B\) over the same concurrent alphabet \((\Sigma, D)\) satisfy the same saturated safety properties if and only if \(TL(\A) = TL(\B)\).
\end{prop}

\begin{proof}
	By Proposition \ref{satsafe}, \(\A\) and \(\B\) satisfy the same saturated safety properties if they have the same trace language. Suppose that \(TL(\A) \not= TL(\B)\). Then we may suppose that there exists an element \(v\in TL(\A)\) such that \(v \notin TL(\B)\). Consider the safety property given by the set \[B = \{x \in \Sigma^* \,|\, v \preceq [x]\}.\]
	This is a saturated safety property. Indeed, let \(m \in B\Sigma^*\) and \(m' \equiv m\). Then there exists an element \(x \in B\) such that \(x \preceq m\). Hence \(v \preceq [x] \preceq [m] = [m']\) and therefore \(m' \in B \subseteq B\Sigma^*\). Note that the same argument shows that \(B = B\Sigma^*\). Since \(v \in TL(\A)\), there exists a path \(\omega \in P_\A^{\mathbb I}\) such that \(\omega(0) = I_\A\) and \(v \preceq [\overline{\lambda}_\A(\omega)]\). Hence \(\overline{\lambda}_\A(\omega) \in B = B\Sigma^*\). Thus \(\A \not \models \mathit{safety}(B)\). On the other hand, \(\B \models \mathit{safety}(B)\). Indeed, if this was not the case, there would exist a path \(\nu \in P_\B^{\mathbb I}\) such that \(\nu(0) = I_\B\) and \(\overline{\lambda}_\B(\nu) \in B\Sigma^* = B\).  But then we would have \(v \preceq [\overline{\lambda}_\B(\nu)]\) and therefore \(v \in TL(\B)\), which is not the case.
\end{proof}

\begin{prop} \label{TLinc}
		Let  \((f,\sigma)\colon \A \to \B\) be a cubical dimap of HDAs.  Then \(M(\sigma)(TL(\A)) \subseteq TL(\B)\). In particular, if \((\Sigma_\A, D_\A) = (\Sigma_\B, D_\B)\) and \(\sigma = id\), then \(TL(\A) \subseteq TL(\B)\).	
\end{prop}

\begin{proof}
	Consider an element \(v \in TL(\A)\). Let \(\omega \in P_\A^{\mathbb I}\) be a path such that \(\omega(0) = I_\A\) and \(v \preceq [\overline{\lambda}_\A(\omega)]\). Let \(w \in M(\Sigma_\A,D_\A)\) such that \([\overline{\lambda}_\A(\omega)] = vw\). Then \(M(\sigma)(v)M(\sigma)(w) = M(\sigma)(vw) = M(\sigma)([\overline{\lambda}_\A(\omega)]) = [\sigma^*(\overline{\lambda}_\A(\omega))] = [\overline{\lambda}_\B(f^{\mathbb I}(\omega))]\). Hence \(M(\sigma)(v) \preceq [\overline{\lambda}_\B(f^{\mathbb I}(\omega))]\). Since \(f^{\mathbb I}(\omega)(0) = f_0(\omega(0)) = f_0(I_\A) = I_\B\), it follows that \(M(\sigma)(v) \in TL(\B)\). 			
\end{proof}

\subsection{The trace language of a tensor product} 

Let \(\A\) and \(\B\) be two HDAs. We view \(M(\Sigma_\A,D_\A)\) and \(M(\Sigma_\B,D_\B)\) as submonoids and \(TL(\A)\) and \(TL(\B)\) as subsets of \(M(\Sigma_{\A\otimes \B},D_{\A \otimes \B})\). 

\begin{prop} \label{TLtensor}
	\(TL(\A \otimes \B) = TL(\A)\cdot TL(\B)\). 
\end{prop}

\begin{proof}
	In order to show that \(TL(\A)\cdot TL(\B)\subseteq TL(\A \otimes \B)\),  consider elements \(a \in TL(\A)\) and \(b \in TL(\B)\). Let \(\alpha \in P_\A ^{\mathbb I}\) and \(\beta \in P_\B^{\mathbb I}\) be paths such that \(\alpha(0) = I_\A\), \(\beta(0) = I_\B\), \(a \preceq [\overline{\lambda}_\A(\alpha)]\), and \(b \preceq [\overline{\lambda}_\B(\beta)]\). Let \(v \in M(\Sigma_\A,D_\A)\) and \(w \in M(\Sigma_\B,D_\B)\) be elements such that \(av = [\overline{\lambda}_\A(\alpha)]\) and \(bw = [\overline{\lambda}_\B(\beta)]\). Write \(\alpha = x_{1\sharp} \cdots x_{k\sharp}\) and \(\beta = y_{1\sharp} \cdots y_{l\sharp}\) where the  elements \(x_i\) and \(y_i\) are of degree \(\leq 1\). Let \(\omega\) be the path in \(\A \otimes \B\) defined by
	\[\omega = (x_1,I_\B)_\sharp \cdots (x_k,I_\B)_\sharp \cdot (\alpha(\length_\alpha),y_1)_\sharp \cdots (\alpha(\length_\alpha),y_l)_\sharp.\]
	Since all elements of \(M(\Sigma_\A,D_\A)\) commute in \(M(\Sigma_{\A\otimes \B},D_{\A \otimes \B})\) with all elements of \(M(\Sigma_\B,D_\B)\), we have \([\overline{\lambda}_{\A\otimes \B}(\omega)] = [\overline{\lambda}_\A(\alpha)\overline{\lambda}_\B(\beta)] = [\overline{\lambda}_\A(\alpha)][\overline{\lambda}_\B(\beta)] = avbw =  abvw\). Hence \(ab \preceq [\overline{\lambda}_{\A\otimes \B}(\omega)]\) and therefore \(ab \in TL(\A\otimes \B)\).
	
	For the reverse inclusion, consider an element \(v \in TL(\A \otimes \B)\). 
	Let \(\omega \in P_{\A \otimes \B}^{\mathbb I}\) be a path such that \(v \preceq [\overline{\lambda}_{\A \otimes \B}(\omega)]\). Consider an element \(w \in M(\Sigma_{\A \otimes B},D_{\A \otimes \B})\) such that \([\overline{\lambda}_{\A \otimes \B}(\omega)] = vw\). Since \(M(\Sigma_{\A\otimes \B},D_{\A \otimes \B}) = M(\Sigma_\A,D_\A)\cdot M(\Sigma_\B,D_\B)\), we may choose elements \(a,x \in M(\Sigma_\A, D_\A)\) and \(b,y \in M(\Sigma_\B,D_\B)\) such that \(v = ab\) and \(w = xy\). We show that \(a \in TL(\A)\) and \(b \in TL(\B)\). Write 
	\(\omega = (x_1,y_1)_{\sharp}\cdots (x_n,y_n)_{\sharp}\), and consider the paths \(\alpha \in P_\A^{\mathbb I}\) and \(\beta \in P_\B^{\mathbb I}\) given by \(\alpha = x_{1\sharp} \cdots x_{n\sharp}\) and \(\beta = y_{1\sharp} \cdots y_{n\sharp}\). Since all elements of \(M(\Sigma_\A, D_\A)\) commute with all elements of \(M(\Sigma_\B, D_\B)\) in \(M(\Sigma_{\A\otimes \B}, D_{\A \otimes \B})\), we have 
	\begin{align*}
	[\overline{\lambda}_{\A\otimes \B} (\omega)] &= [\overline{\lambda}_{\A\otimes \B}((x_1,y_1)_{\sharp})]\cdots [\overline{\lambda}_{\A\otimes \B}((x_n,y_n)_{\sharp})]\\
	&= [\overline{\lambda}_{\A}(x_{1\sharp})][\overline{\lambda}_{\B}(y_{1\sharp})]\cdots [\overline{\lambda}_{\A}(x_{n\sharp})][\overline{\lambda}_{\B}(y_{n\sharp})]\\
	&= [\overline{\lambda}_{\A}(x_{1\sharp})]\cdots [\overline{\lambda}_{\A}(x_{n\sharp})][\overline{\lambda}_{\B}(y_{1\sharp})] \cdots [\overline{\lambda}_{\B}(y_{n\sharp})]\\
	&= [\overline{\lambda}_{\A}(\alpha)][\overline{\lambda}_{\B}(\beta)].
	\end{align*}
	Hence \(axby = abxy = vw = [\overline{\lambda}_{\A\otimes \B} (\omega)] = [\overline{\lambda}_{\A}(\alpha)][\overline{\lambda}_{\B}(\beta)]\). Since every element of \(M(\Sigma_{\A\otimes \B}, D_{\A \otimes \B})\) can be uniquely written as a product of an element of \(M(\Sigma_\A, D_\A)\) and an element of \(M(\Sigma_\B, D_\B)\), it follows that \(ax = [\overline{\lambda}_{\A}(\alpha)]\) and \(by = [\overline{\lambda}_{\B}(\beta)]\). Thus \(a \preceq  [\overline{\lambda}_{\A}(\alpha)]\) and \(b \preceq [\overline{\lambda}_{\B}(\beta)]\) and therefore \(a \in TL(\A)\) and \(b \in TL(\B)\). 
\end{proof}

\subsection{Fundamental monoid} The \emph{fundamental monoid} of an HDA \(\A\) is the submonoid \(\pi(\A)\) of \(M(\Sigma_\A, D_\A)\) defined by
\[\pi(\A) = \{[\overline{\lambda}_\A(\omega)] \,|\, \omega \in P_\A^{\mathbb I},\, \omega(0) = \omega(\length_\omega) = I_\A\}.\]
The term reflects an analogy with the fundamental group of a topological space. Given a cubical dimap \((f,\sigma) \colon \A \to \B\), the homomorphism \(M(\sigma) \colon M(\Sigma_\A, D_\A) \to M(\Sigma_\B, D_\B)\) restricts to a homomorphism \(\pi(f,\sigma) \colon \pi(\A) \to \pi(\B)\). In particular, we have the following proposition:

\begin{prop} \label{piinc}
	Let \(\A\) and \(\B\) be two HDAs over the same concurrent alphabet. If there exists a cubical dimap of HDAs \((f,\sigma ) \colon \A \to \B\) such that \(\sigma = id\), then \(\pi(\A)\) is a submonoid of \(\pi(\B)\). 
\end{prop}

\begin{prop} \label{pitensor}
	For any two HDAs \(\A\) and \(\B\), \(\pi(\A \otimes \B) = \pi(\A) \cdot \pi(\B)\). 
\end{prop}

\begin{proof}
	The proof is similar to the one of Proposition \ref{TLtensor}. The details are left to the reader.
\end{proof}

\begin{prop} \label{picoprod}	
		Let \(\A\) and \(\B\) be two HDAs. The isomorphism of monoids \[{M(\Sigma_\A,D_\A)\ast M(\Sigma_\B,D_\B) \to M(\Sigma_{\A+\B},D_{\A+\B})}\] restricts to an isomorphism \({\pi(\A)\ast \pi(\B) \to \pi(\A + \B)}\). 
\end{prop} 

\begin{proof}
	Let \((j_\A,\sigma_\A)\colon \A \to \A +\B\) and \((j_\B, \sigma_\B) \colon \B \to {\A+\B}\) be the morphisms of HDAs where the morphisms of precubical sets \(j_\A \colon P_\A \to P_{\A+\B}\) and \(j_\B \colon P_\B \to P_{\A+\B}\) are given by \(j_\A(x) = (x,I_\B)\) and \(j_\B(y) = (I_\A,y)\) and \(\sigma_\A \colon (\Sigma_\A,D_\A) \to (\Sigma_{\A+\B},D_{\A+\B})\) and \(\sigma_\B \colon (\Sigma_\B,D_\B) \to (\Sigma_{\A+\B},D_{\A+\B})\) are the inclusions. Consider the following commutative diagram of monoids:
	\[
	\begin{tikzcd}[column sep=6pc]
	\pi(\A)\ast \pi(\B) \ar[r, "(\pi(|j_\A|{,}\sigma_\A){,}\pi(|j_\B|{,}\sigma_\B))"] \ar[d, hook] & \pi(\A +\B) \ar[d, hook]\\
	M(\Sigma_\A, D_\A) \ast M(\Sigma_\B, D_\B) \ar[r, "\cong",  "(M(\sigma_\A){,}M(\sigma_\B))"'] & M(\Sigma_{\A+\B}, D_{\A+\B})
	\end{tikzcd}	
	\]
	It is clear that \((\pi(|j_\A|{,}\sigma_\A){,}\pi(|j_\B|{,}\sigma_\B))\) is injective. We show that it is surjective. A loop \(\omega\) in \(\A +\B\) with \(\omega(0) = I_{\A+\B} = (I_\A,I_\B)\) can be decomposed as a concatenation
	\[\omega = (j_\A \circ \alpha_1) (j_\B\circ \beta_1)\cdots (j_\A \circ \alpha_n) (j_\B\circ \beta_n)\]
	where the \(\alpha_i\) are loops in \(\A\) with \(\alpha_i(0) = I_\A\) and the \(\beta_i\) are loops in \(\B\) with \(\beta_i(0) = I_\B\). Therefore
	\begin{align*}
	\MoveEqLeft{[\overline{\lambda}_{\A+\B}(\omega)]}\\ &= 	[\overline{\lambda}_{\A+\B}(j_\A\circ \alpha_1)] 	[\overline{\lambda}_{\A+\B}(j_\B\circ \beta_1)]\cdots 	[\overline{\lambda}_{\A+\B}(j_\A\circ \alpha_n)]	[\overline{\lambda}_{\A+\B}(j_\B\circ \beta_n)]\\
	&= 	[\overline{\lambda}_{\A+\B}(|j_\A|^{\mathbb I}(\alpha_1))] 	[\overline{\lambda}_{\A+\B}(|j_\B|^{\mathbb I}(\beta_1))]\cdots 	[\overline{\lambda}_{\A+\B}(|j_\A|^{\mathbb I}(\alpha_n))]	[\overline{\lambda}_{\A+\B}(|j_\B|^{\mathbb I}(\beta_n))]\\
	&= 	[\sigma_\A^*(\overline{\lambda}_{\A}(\alpha_1))] 	[\sigma_\B^*(\overline{\lambda}_{\B}(\beta_1))]\cdots 	[\sigma_\A^*(\overline{\lambda}_{\A}(\alpha_n))]	[\sigma_\B^*(\overline{\lambda}_{\B}(\beta_n))]\\
	&= 	M(\sigma_\A)([\overline{\lambda}_{\A}(\alpha_1)])M(\sigma_\B)([
	\overline{\lambda}_{\B}(\beta_1)])\cdots 	M(\sigma_\A)([\overline{\lambda}_{\A}(\alpha_n)])M(\sigma_\B)([\overline{\lambda}_{\B}(\beta_n)])\\
	&= 	(M(\sigma_\A),M(\sigma_\B))([\overline{\lambda}_{\A}(\alpha_1)][
	\overline{\lambda}_{\B}(\beta_1)]\cdots 	[\overline{\lambda}_{\A}(\alpha_n)][\overline{\lambda}_{\B}(\beta_n)])\\
	&= 	(\pi(|j_\A|,\sigma_\A),\pi(|j_\B|,\sigma_\B))([\overline{\lambda}_{\A}(\alpha_1)][
	\overline{\lambda}_{\B}(\beta_1)]\cdots 	[\overline{\lambda}_{\A}(\alpha_n)][\overline{\lambda}_{\B}(\beta_n)]). 
	\end{align*}
	It follows that \((\pi(|j_\A|{,}\sigma_\A){,}\pi(|j_\B|{,}\sigma_\B))\) is surjective.   	
\end{proof}

\subsection{The trace language of a coproduct}

Let \(\A\) and \(\B\) be two HDAs. We view \(M(\Sigma_\A,D_\A)\) and \(M(\Sigma_\B,D_\B)\) as submonoids and \(TL(\A)\) and \(TL(\B)\) as subsets of \(M(\Sigma_{\A+ \B},D_{\A + \B})\). Similarly, we view paths in \(\A\) and paths in \(\B\) as paths in \(\A+\B\). 

\begin{prop} \label{TLcoprod}
	\(TL(\A+\B) = \pi(\A+\B)\cdot TL(\A) \cup \pi(\A+\B)\cdot TL(\B)\).
\end{prop}

\begin{proof}
	We show first that \(\pi(\A+\B)\cdot TL(\A) \cup \pi(\A+\B)\cdot TL(\B) \subseteq TL(\A+\B)\). Consider an element \(v \in TL(\A)\) and a loop \(\omega\in P_{\A + \B}^{\mathbb I}\) such that \(\omega(0) = I_{\A+\B} = (I_\A,I_\B)\). Let \(\alpha\) be a path in \(\A\) such that \(\alpha(0) = I_\A\) and \(vw = [\overline{\lambda}_\A(\alpha)]\) for some element \(w \in M(\Sigma_\A,D_\A)\). We have \([\overline{\lambda}_{\A+\B}(\omega\cdot \alpha)] = [\overline{\lambda}_{\A+\B}(\omega)]vw\). Hence \([\overline{\lambda}_{\A+\B}(\omega)]v \in TL(\A+\B)\). Thus \(\pi(\A+\B)\cdot TL(\A) \subseteq TL(\A+\B)\). Similarly, \(\pi(\A+\B)\cdot TL(\B) \subseteq TL(\A+\B)\).
	
	For the reverse inclusion, consider an element \(v \in TL(\A+\B)\). Let \(\omega \in P_{\A+\B}^{\mathbb I}\) be a path of minimal length such that \(\omega(0) = I_{\A+\B}\) and \([\overline{\lambda}_{\A+\B}(\omega)] = vw\) for some element \(w \in M(\Sigma_{\A +\B}, D_{\A+\B})\). We may assume  that \(\omega\) does not lie entirely in \(\A\) or \(\B\), because in that case we would have either \(v \in TL(\A)\) or \(v \in TL(\B)\) and there would be nothing to prove. We may further suppose that the last edge of \(\omega\) is an edge of \(\A\) and leave the analogous case where it is an edge of \(\B\) to the reader. Decompose \(\omega\) as a concatenation
	\[ \omega = \alpha_1 \cdot \beta_1 \cdots \alpha_r \cdot \beta_r \cdot \gamma\]
	where the \(\alpha_i\) are loops in \(\A\) with \(\alpha_i(0) = I_\A\), the \(\beta_i\) are loops in \(\B\) with \(\beta_i(0) = I_\B\), \(\gamma\) is a path in \(\A\) with \(\gamma (0) = I_\A\), and all paths except possibly \(\alpha_1\) have positive length. Set \(l_1 = \length_{\alpha_1\cdot \beta_1 \cdots \alpha_r\cdot \beta_r}\) and \(l_2 = \length_\gamma\). Then \(|v| > l_1\). Indeed, otherwise we would have \(w = yz\) for some elements \(y,z \in  M(\Sigma_{\A +\B}, D_{\A+\B})\) such that \(|y| = l_1 - |v|\) and \(|z| = l_2\). Moreover, we would have
	\[vyz = [\overline{\lambda}_\A(\alpha_1)][\overline{\lambda}_\B(\beta_1)] \cdots [\overline{\lambda}_\A(\alpha_r)][\overline{\lambda}_\B(\beta_r)][\overline{\lambda}_\A(\gamma)].\]
	Since congruence classes of elements of \(\Sigma_\A\) do not commute with congruence classes of elements of \(\Sigma_\B\) 
	in \(M(\Sigma_{\A+\B},D_{\A+\B})\), this would imply 
	\(z = [\overline{\lambda}_\A(\gamma)]\) and 
	\[vy = [\overline{\lambda}_\A(\alpha_1)][\overline{\lambda}_\B(\beta_1)] \cdots [\overline{\lambda}_\A(\alpha_r)][\overline{\lambda}_\B(\beta_r)] = [\overline{\lambda}_{\A+\B}(\alpha_1 \cdot \beta_1 \cdots \alpha_r \cdot \beta_r)],\] 
	which would contradict the minimality of \(\omega\). So \(|v| > l_1\). Hence there exist elements \[u, x \in M(\Sigma_{\A +\B}, D_{\A+\B})\] such that \(|u| = l_1\), \(|x| = |v| -l_1\), and \(v = ux\). Since 
	\[uxw = [\overline{\lambda}_\A(\alpha_1)][\overline{\lambda}_\B(\beta_1)] \cdots [\overline{\lambda}_\A(\alpha_r)][\overline{\lambda}_\B(\beta_r)][\overline{\lambda}_\A(\gamma)],\]
	we have \[u = [\overline{\lambda}_\A(\alpha_1)][\overline{\lambda}_\B(\beta_1)] \cdots [\overline{\lambda}_\A(\alpha_r)][\overline{\lambda}_\B(\beta_r)] = [\overline{\lambda}_{\A+\B}(\alpha_1 \cdot \beta_1 \cdots \alpha_r \cdot \beta_r)]\] and \(xw = [\overline{\lambda}_\A(\gamma)]\). Therefore \(u \in \pi(\A + \B)\), \(w \in M(\Sigma_\A,D_\A)\), \(x \in TL(\A)\), and \(v = ux \in \pi(\A + \B)\cdot TL(\A)\). 	
\end{proof}

\begin{rem}
	By Propositions \ref{picoprod} and \ref{TLcoprod}, we have \(TL(\A + \B) = TL(\A'+\B')\) if \((\Sigma_\A,D_\A) = (\Sigma_{\A'},D_{\A'})\), \((\Sigma_\B,D_\B) = (\Sigma_{\B'},D_{\B'})\), \(TL(\A) = TL(\A')\), \(TL(\B) = TL(\B')\), \(\pi(\A) = \pi(\A')\), and \(\pi(\B) = \pi(\B')\). The last two assumptions are needed here, as shows the example where \(\A\) and \(\B\) have only one vertex and one edge, the one of \(\A\) labeled \(a\) and the one of \(\B\) labeled \(b\), \(\A' = \A\), and \(\B'\) has two vertices and two edges, both labeled \(b\), one leading from the initial to the other state and the other leading from the second state to itself.  
\end{rem}

\section{The homology language of an HDA} \label{secHL}

A higher-dimensional automaton is an ordinary automaton with information on independence of actions. We have used the independence relation associated with the concurrent alphabet of an HDA, and the induced congruence relation, to define its trace language and its fundamental monoid. The higher-dimensional structure of an HDA contains further information on independence. An overall picture of the independence structure of an HDA is given by its labeled homology, as introduced in \cite{labels}. Here, we use the labeling on the homology of an HDA to define its homology language. As in the case of the trace language and the fundamental monoid, we show that the homology language is compatible with cubical dimaps and establish formulas to compute it for tensor products and coproducts. We also give examples of how the homology language of an HDA can be used to reason about the independence of subsystems or components of a concurrent system. We work over a fixed principal ideal domain, which we suppress from the notation. 

\subsection{Chain complexes and homology} A \emph{chain complex} is a graded module \(C = (C_n)_{n\geq 0}\) with  \emph{boundary operators} \(d\colon C_n \to C_{n-1}\) \((n \geq 1)\) satisfying \(d\circ d = 0\). A \emph{chain map} between two chain complexes is a morphism of graded modules that commutes with the boundary operators. The \emph{homology} of a chain complex \(C\) is the graded module \(H_*(C) = (H_n(C))_{n\geq 0}\) defined by \(H_0(C) = C_0/\im (d \colon C_{1} \to C_0)\) and 
\[H_n(C) = \ker (d \colon C_n \to C_{n-1})/ \im (d \colon C_{n+1} \to C_n) \quad (n\geq 1).\] A chain map \(f\colon C \to D\) induces a morphism of graded modules \(f_*\colon H_*(C) \to H_*(D)\), defined by \(f_*([z]) = [f(z)]\), and this makes \(H_*\) a functor from the category of chain complexes to the category of graded modules. 

The \emph{direct sum} of two chain complexes is the direct sum of the underlying graded modules, with boundary operators defined componentwise. The homology functor preserves direct sums.
The \emph{tensor product} of two graded modules \(A\) and \(B\) is the graded module \(A\otimes B\) defined by \[(A\otimes B)_n = \bigoplus \limits_{0 \leq i \leq n} A_i\otimes B_{n-i}.\] The \emph{tensor product} of two chain complexes \(C\) and \(D\) is the tensor product of the underlying graded modules with the boundary operators given by 
\[d(x\otimes y) = dx \otimes y + (-1)^i x \otimes dy, \quad x \in C_i, \, y \in D_{n-i}.\]
Over a field, the homology functor is compatible with tensor products. For the general case and further results in homological algebra, see, \textit{e.g.}, \cite{Dold, Hatcher}.

\subsection{Cubical chains and cubical homology} Let \(P\) be a precubical set. The \emph{cubical chain complex} of \(P\) is the  chain complex \(C_*(P)\) where \(C_n(P)\) is the free module generated by \(P_n\) and the boundary operator \(d\colon C_n(P) \to C_{n-1}(P)\) is given by \[dx = \sum \limits_{i=1}^{n}(-1)^i(d^0_ix -d^1_ix), \quad x \in P_n.\]  
The chain map induced by a morphism of precubical sets is defined in the obvious way. The \emph{cubical homology} of \(P\), denoted  by \(H_*(P)\), is the homology of \(C_*(P)\). The cubical chain complex $C_*(P)$ is naturally isomorphic to the cellular chain complex of $|P|$ (cf. \cite[Thm. 3.3.1]{labels}). Since a cubical dimap of precubical sets is a cellular map, it follows that the functors $C_*$ and $H_*$ extend to the category of precubical sets and cubical dimaps and, moreover, that a cubical dimap which is a homotopy equivalence induces an isomorphism in cubical homology. An explicit description of the chain map induced by an elementary cubical dimap is given in \cite[Prop. 7.4.1]{labels}. 

\begin{ex} \label{toy}
	Throughout this section, we will consider the example 
	HDA \(\A\) where \((P_\A)_0 = \{I_\A\}\), \((P_\A)_1 = \{x_1,x_2,x_3\}\), \((P_\A)_2= \{y\}\), \((P_\A)_n = \emptyset\)  \((n \geq 3)\), \(d^k_1x_i  = I_\A\) \((k \in \{0,1\}, i \in \{1,2,3\})\),  \(d_1^ky = x_2\), \(d_2^ky = x_1\) \((k \in \{0,1\})\), \(F_\A = \{I_\A\}\), \(\Sigma_\A = \{a_1,a_2,a_3\}\), \(D_\A = (\Sigma_\A \times \Sigma_\A) \setminus  \{(a_1,a_2),(a_2,a_1)\}\), and  \(\lambda_\A(x_i) = a_i\) \(({i \in \{1,2,3\}})\). We suppose, of course, that the \(x_i\) and the \(a_i\) are pairwise different. Geometrically, \(\A\) is a wedge (one-point union) of a torus and a circle. By definition of the cubical chain complex, \(C_0(P_\A)\) is the free module generated by \(I_\A\),  \(C_1(P_{\A})\) is the free module generated by the \(x_i\), \(C_2(P_{\A})\) is the free module generated by
	\(y\), and all other \(C_n(P_{\A})\) are \(0\). Since \(d^0_1x_i = d^1_1x_i\) and \(d^0_iy = d^1_iy\), all boundary operators of \(C_*(P_\A)\) are \(0\). Hence \(H_*(P_\A)\) has \(1\) generator in degrees \(0\) and \(2\) and \(3\) generators in degree \(1\). As this example illustrates, homology may be seen as an algebraic tool to count holes in geometric objects such as precubical sets or topological spaces. 	
\end{ex}

\subsection{The edge \texorpdfstring{\(e_ix\)}{}} \label{eki}

Let \(x\) be an element of degree \(n > 0\) of a precubical set \(P\), and let \(i \in \{1, \dots, n\}\). We define the \emph{\(i\)th starting edge} of \(x\) to be the element \(e_ix \in P_1\) given by 
\[e_ix = \left\{\begin{array}{ll} x,& n = 1,\\d_1^{0} \cdots d_{i-1}^ {0}d_{i+1}^{0}\cdots d_n^{0}x, & n > 1. \end{array} \right.\]  
The edge \(e_ix\) leads from the initial vertex of \(x\) to the initial vertex of the face \(d^1_ix\), \textit{i.e.}, we have \(d_1^ 0e_ix = {d_1^ 0\cdots d_1^0x}\) and \(d_1^1e_ix = d_1^ 0\cdots d_1^0d_i^ 1x\). An illustration is given in Figure \ref{edges}.  

\begin{figure}
	\center
	\begin{tikzpicture}[initial text={},on grid]

	\path[draw, fill=lightgray] (0,1)--(1,1)--(1,0)--(0,0)--cycle;

	\node[state,minimum size=0pt,inner sep =2pt,fill=white] (p_0) at (0,0)  {}; 
	
	\node[state,minimum size=0pt,inner sep =2pt,fill=white,label={[label distance = 0.3cm]315:\scalebox{0.85}{$x$}}] (q_0) [above=of p_0,xshift=0cm] {}; 	   
	
	\node[state,minimum size=0pt,inner sep =2pt,fill=white] (p_2) [right=of p_0,xshift=0cm] {};
	
	\node[state,minimum size=0pt,inner sep =2pt,fill=white] (p_1) [above=of p_2,xshift=0cm] {};

	\path[->] 
	(q_0) edge[above] node {\scalebox{0.85}{$e_1x = d^0_2x$}} (p_1)
	(q_0) edge[left]  node {\scalebox{0.85}{$e_2x = d^0_1x$}} (p_0)
	(p_1) edge[above] node {} (p_2)
	(p_0) edge[below] node {} (p_2)
	;
	\end{tikzpicture}
	\caption{A $2$-cube and its starting edges }\label{edges}
\end{figure}
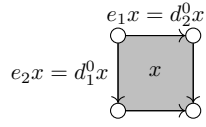

\subsection{Strings} Let \(\Sigma\) be an alphabet. Given a string \(m\) of length \(n\geq 1\), we will write \(m_1, \dots, m_n\) to denote the uniquely determined elements of \(\Sigma\) such that \(m = m_1\cdots m_n\).

\subsection{Labeling chain map}

Let \(\A\) be an HDA. Consider the exterior algebra on the free module generated by \(\Sigma_\A\), \(\Lambda (\Sigma_\A)\). Recall that this is the quotient of the tensor algebra on the free module on \(\Sigma_\A\) by the two-sided ideal generated by all elements of the form \(x\otimes x\) where \(x \in \Sigma_\A\) (see \cite{BourbakiAlgI} for more details). The exterior algebra \(\Lambda (\Sigma_\A)\) is canonically graded by the exterior powers of the free module generated by \(\Sigma_\A\). We view the graded module \(\Lambda(\Sigma_\A)\) as a chain complex with \(d=0\) and define the \emph{labeling chain map} \[\mathfrak{l}_{\A} \colon C_*(P_\A) \to \Lambda (\Sigma_\A)\] on basis elements \(x \in (P_\A)_n\) by 
\[\mathfrak{l}_{\A}(x) = \left \{\begin{array}{ll}
1_{\Lambda(\Sigma_\A)}, & n = 0, \vspace{0.2cm}\\
\sum \limits _{j_1 = 1}^{|\lambda_{\A}(e_1x)|} \dots \sum \limits _{j_n = 1}^{|\lambda_{\A}(e_nx)|} \lambda_{\A}(e_1x)_{j_1}\wedge \dots \wedge  \lambda_{\A}(e_nx)_{j_n},&  n > 0.
\end{array} \right.\]
By \cite[Prop. 4.4.5]{labels}, the labeling chain map is indeed a chain map, i.e, we have \(\mathfrak{l}_\A(dx) = d\mathfrak{l}_\A(x) = 0\) for all \(x \in C_*(P_\A)\). 

\begin{ex} \label{toylabels}
	Consider the HDA \(\A\) of Example \ref{toy}. The exterior algebra \(\Lambda (\Sigma_\A)\) is the graded module freely generated by \(1_{\Lambda(\Sigma_\A)}\) in degree 0, \(a_1\), \(a_2\), and \(a_3\) in degree \(1\), \(a_1 \wedge a_2\), \(a_1\wedge a_3\), and \(a_2\wedge a_3\) in deegre \(2\), and \(a_1\wedge a_2\wedge a_3\) in degree \(3\). In degrees \(\geq 4\), \(\Lambda (\Sigma_\A)\) is \(0\). We have \(e_1x_i = x_i\), \(e_1y = d^0_2y = x_1\), and \(e_2y = d^0_1y = x_2\). Hence the labeling chain map of \(\A\) is given by \(\mathfrak{l}_\A(I_\A) = 1_{\Lambda(\Sigma_\A)}\), \(\mathfrak{l}_\A(x_1) = a_1\), \(\mathfrak{l}_\A(x_2) = a_2\), \(\mathfrak{l}_\A(x_3) = a_3\), and \(\mathfrak{l}_\A(y) = a_1\wedge a_2\).
\end{ex}

\begin{prop} \label{fraknat}
	Let \((f, \sigma) \colon \A \to \B\) be a cubical dimap of HDAs, and let \(f_*\colon C_*(P_\A) \to C_*(P_\B)\) be the chain map induced by $f$.  Then \(\mathfrak{l}_\B \circ f_* = \Lambda(\sigma) \circ \mathfrak{l}_\A\).
\end{prop}

\begin{proof}
	We may suppose that $(f,\sigma)$ is an elementary cubical dimap of HDAs. Consider the HDA \(\C\) given by \(P_\C = P_\A\), \(I_\C = I_\A\), \(F_\C = F_\A\), \(\Sigma_\C = \Sigma_\B\), \(D_\C = D_\B\), and \(\lambda_\C = \sigma^*\circ \lambda_\A\). Then \((f,\sigma)\) decomposes as the composite of elementary cubical dimaps of HDAs 
	\[\A \xrightarrow{(id_{|P_\A|}, \sigma)} \C \xrightarrow{(f,id_{(\Sigma_\B,D_\B)})}  \B.\]
	We have \(\mathfrak{l}_\C = \Lambda(\sigma) \circ \mathfrak{l}_\A\) and, by \cite[Thm. 7.5.1]{labels}, \(\mathfrak{l}_\B \circ f_* =  \mathfrak{l}_\C\). Hence \(\mathfrak{l}_\B \circ f_* = \Lambda(\sigma) \circ \mathfrak{l}_\A\).	
\end{proof}

\subsection{Labeled homology}

Let \(\A\) be an HDA. The labeling chain map \(\mathfrak{l}_{\A}\) induces 
a morphism of graded modules
\[ \ell_{\A}\colon H_*(P_\A) \to H_*(\Lambda (\Sigma_\A)) \cong \Lambda (\Sigma_\A).\]
Explicitly, \(\ell_\A([z]) = \mathfrak{l}_\A(z)\). The pair \((H_*(P_\A),\ell_\A)\) is called the \emph{labeled homology} of \(\A\). 

\begin{ex} \label{toyell}
	For the HDA \(\A\) considered in Examples \ref{toy} and \ref{toylabels}, we have \(\ell_\A([I_\A]) = 1_{\Lambda(\Sigma_\A)}\), \(\ell_\A([x_i]) = a_i\) \((i \in \{1,2,3\})\), and \(\ell_\A([y]) = a_1\wedge a_2\). Further  examples can be found in \cite{labels}.
\end{ex}

Proposition \ref{fraknat} immediately implies the following fact:

\begin{prop} \label{ellnat}
	Let \((f, \sigma) \colon \A \to \B\) be a cubical dimap of HDAs. Then the morphism of graded modules \(f_*\colon H_*(P_\A) \to H_*(P_\B)\) satisfies \(\ell_\B \circ f_* = \Lambda(\sigma) \circ \ell_\A\).
\end{prop}

\subsection{The homology language}

We define the \emph{homology language} of an HDA \(\A\) to be the graded module
\[HL(\A) = \im \; (\ell_\A \colon H_*(P_\A) \to \Lambda(\Sigma_\A)).\]
Thus, by definition, the homology language of an HDA can be read off its labeled homology.

\begin{ex} \label{toyHL}
	The homology language of the HDA \(\A\) of Examples \ref{toy}, \ref{toylabels}, and \ref{toyell} is the graded submodule of \(\Lambda (\Sigma_\A)\) generated by the unit and the elements \(a_1\), \(a_2\), \(a_3\), and \(a_1\wedge a_2\).
\end{ex}

\begin{prop} \label{HLinc}
	Let \((f, \sigma) \colon \A \to \B\) be a cubical dimap of HDAs. Then \(\Lambda(\sigma) (HL(\A)) \subseteq HL(\B)\). In particular, if \((\Sigma_\A, D_\A) = (\Sigma_\B, D_\B)\) and \(\sigma = id\), then \(HL(\A) \subseteq HL(\B)\).	If, furthermore, \(f\) is a homotopy equivalence, then \(HL(\A) = HL(\B)\).
\end{prop}

\begin{proof}
	This follows from Proposition \ref{ellnat} and the fact that a homotopy equivalence induces an isomorphism in homology. 
\end{proof}

\subsection{The homology language of a tensor product}

Let \(\A\) and \(\B\) be two HDAs. We view \(HL(\A)\) and \(HL(\B)\) as graded submodules of the exterior algebra \(\Lambda(\Sigma_{\A\otimes \B}) = \Lambda(\Sigma_\A \amalg \Sigma_\B)\). 

\begin{prop} \label{HLtensor}
	\(HL(\A \otimes \B) = HL(\A) \wedge HL(\B)\).
\end{prop}

\begin{proof}
	Consider the homology cross product
	\[ \times \colon  H_*(P_\A) \otimes H_*(P_\B) \to  H_*(P_\A\otimes P_\B) = H_*(P_{\A\otimes \B}),\]
	\textit{i.e.}, the composite \(\zeta_*\circ \kappa\) where \(\kappa\) is the homomorphism of graded modules 
	\[H_*(P_\A)\otimes H_*(P_\B) \to H_*(C_*(P_\A)\otimes C_*(P_\B)), \quad [x]\otimes [y] \mapsto [x\otimes y]\] 
	and \(\zeta\) is the isomorphism of chain complexes \(C_*(P_\A)\otimes C_*(P_\B) \to C_*(P_\A\otimes P_\B)\) given by \[x\otimes y \mapsto (x,y), \quad x\in P_\A,\, y\in P_\B.\]
	By \cite[Thm. 5.3.2]{labels}, we have the following commutative diagram of graded modules:	
	\[
	\begin{tikzcd}
	H_*(P_\A) \otimes H_*(P_\B) \ar[r, "\times"] \ar[d, "\ell_\A\otimes \ell_\B"' ] &  H_*(P_\A\otimes P_\B) \ar[r, "="] & H_*(P_{\A\otimes \B}) \ar[d, "\ell_{\A\otimes \B}"]\\
	\Lambda(\Sigma_\A) \otimes  \Lambda(\Sigma_\B) \ar[r, hook] & \Lambda(\Sigma_{\A\otimes \B}) \otimes  \Lambda(\Sigma_{\A\otimes \B}) \ar[r, "\wedge"'] & \Lambda(\Sigma_{\A\otimes \B})
	\end{tikzcd}	
	\]
	By the Künneth theorem, there exists a graded torsion module \(U\subseteq H_*(P_\A\otimes P_\B)\) such that \[{H_*(P_\A\otimes P_\B)} = U \oplus \im \times.\] Since \(\Lambda(\Sigma_{\A \otimes \B})\) is a free module, \(\ell_{\A\otimes \B}(U) = 0\). Hence  \[HL(\A \otimes \B) = \ell_{\A\otimes \B}(U \oplus \im \times) = \ell_{\A\otimes \B} (\im \times) = \im (\ell_{\A\otimes \B} \circ \times).\] By the commutativity of the above diagram, \(\im (\ell_{\A\otimes \B} \circ \times) =  HL(\A) \wedge HL(\B)\). Thus \(HL(\A \otimes \B)  =  HL(\A) \wedge HL(\B)\).
\end{proof}

\begin{ex} \label{toytensor}
	Let \(\A_i\) be the sub-HDA of the HDA \(\A\) of Example \ref{toy} defined by \((P_{\A_i})_0 = \{I_\A\}\), \((P_{\A_i})_1 = \{x_i\}\), and \(\Sigma_{\A_i} = \{a_i\}\). Geometrically, \(\A_i\) is a circle, and one easily computes that \(HL(\A_i) = \Lambda(\{a_i\})\). By Proposition \ref{HLtensor}, for \(i \not=j\), \(HL(\A_i\otimes \A_j) = HL(\A_i)\wedge HL(\A_j) = \Lambda (\{a_i,a_j\})\). 
\end{ex}

\subsection{The homology language of a coproduct}

Let \(\A\) and \(\B\) be two HDAs. We view \(HL(\A)\) and \(HL(\B)\) as graded submodules of the exterior algebra \(\Lambda(\Sigma_{\A+\B}) = \Lambda(\Sigma_\A \amalg \Sigma_\B)\). 

\begin{prop} \label{HLcoprod}
	\(HL(\A + \B) = HL(\A) + HL(\B)\). 
\end{prop}

\begin{proof}
	Let \((j_\A,\sigma_\A)\colon \A \to \A +\B\) and \((j_\B, \sigma_\B) \colon \B \to {\A+\B}\) be the morphisms of HDAs where the morphisms of precubical sets \(j_\A \colon P_\A \to P_{\A+\B}\) and \(j_\B \colon P_\B \to P_{\A+\B}\) are given by \(j_\A(x) = (x,I_\B)\) and \(j_\B(y) = (I_\A,y)\) and \(\sigma_\A \colon (\Sigma_\A,D_\A) \to (\Sigma_{\A+\B},D_{\A+\B})\) and \(\sigma_\B \colon (\Sigma_\B,D_\B) \to (\Sigma_{\A+\B},D_{\A+\B})\) are the inclusions. Consider the morphisms of precubical sets \(I \mapsto I_\A\) and \(I\mapsto I_\B\) from a precubical set with one vertex \(I\) to \(P_\A\) and \(P_\B\), respectively. Then we have the following push out of precubical sets:
	\[
	\begin{tikzcd}
	\{I\} \ar[r] \ar[d] & P_\A \ar[d, "j_\A"] \\
	P_\B \ar[r, "j_\B"'] &P_{\A+\B}
	\end{tikzcd}
	\] 
	Applying cubical chains to this push out, we obtain a push out of chain complexes, which yields a short exact sequence
	\[0 \to C_*(\{I\}) \to C_*(P_\A) \oplus C_*(P_\B) \xrightarrow{(j_{\A*},j_{\B*})} C_*(P_{\A+\B}) \to 0.\]
	The induced long exact sequence in homology shows that the upper map
	in the following commutative diagram of graded modules is surjective: 
	\[
	\begin{tikzcd}[column sep=4.5pc]
	H_*(P_\A) \oplus H_*(P_\B) \ar[r,"(j_{\A*}{,}j_{\B*})"] \ar[d, "\ell_\A \oplus \ell_\B"'] & H_*(P_{\A+\B}) \ar[d, "\ell_{\A+\B}"]\\
	\Lambda(\Sigma_\A) \oplus \Lambda(\Sigma_\B) \ar[r, "(\Lambda(\sigma_\A){,}\Lambda(\sigma_\B))"'] & \Lambda(\Sigma_{\A+ \B})
	\end{tikzcd}
	\]
	Hence \(HL(\A+\B) = \im (\ell_{\A+\B} \circ (j_{\A*}{,}j_{\B*})) = \im ((\Lambda(\sigma_\A){,}\Lambda(\sigma_\B))\circ (\ell_\A \oplus \ell_\B)) = HL(\A) + HL(\B)\).
\end{proof}

\begin{ex}
	Consider the HDA \(\A\) of Example \ref{toy} and the sub-HDAs \(\A_i\) of \(\A\) defined in Example \ref{toytensor}. By Proposition \ref{HLcoprod} and Examples \ref{toyHL} and \ref{toytensor}, we have \(HL((\A_1\otimes \A_2) + \A_3) = {HL(\A_1\otimes \A_2) + HL(\A_3)} = \Lambda (\{a_1,a_2\}) + \Lambda (\{a_3\}) = HL(\A)\), which is no surprise since \(\A \cong (\A_1 \otimes \A_2) + \A_3\). 
\end{ex}

\subsection{Independence} Let \(\A\) be an HDA, and let \(\A_1, \dots, \A_n\) \((n \geq 2)\) be HDAs with disjoint alphabets, each contained in $\Sigma_\A$. We say that the \(\A_i\) are \emph{independent in $\A$} if there exist a sub-HDA \(\B \subseteq \A\) and an isomorphism \((f,\sigma) \colon \A_1 \otimes \cdots \otimes \A_n \to \B\) in the category of HDAs and cubical dimaps such that \(\Sigma_\B = \bigcup_{i=1}^n \Sigma_{\A_i}\) and \(\sigma\) is induced by the inclusions \(\Sigma_{\A_i} \hookrightarrow \Sigma_{\B}\). Since \(\Sigma_{\A_i} \subseteq \Sigma_\B \subseteq \Sigma_\A\), we may view \(HL(\A_i)\) and \(HL(\B)\) as graded submodules of both \(\Lambda(\Sigma_\B)\) and \(\Lambda(\Sigma_\A)\). 

\begin{prop} \label{HLindep}
	If the HDAs \(\A_1, \dots, \A_n\) are independent in \(\A\), then \[HL(\A_1)\wedge \cdots \wedge HL(\A_n) \subseteq HL(\A).\]
\end{prop}

\begin{proof}
	It follows from Proposition \ref{HLtensor} that \(HL(\B) = HL(\A_1)\wedge \cdots \wedge HL(\A_n)\). By Proposition \ref{HLinc}, \(HL(\B) \subseteq HL(\A)\).	
\end{proof}

\begin{exs} \label{secex}

		(i) Consider again the HDA \(\A\) of Example \ref{toy} and the sub-HDAs \(\A_i\) of \(\A\) defined in Example \ref{toytensor}.
		We have \(HL(\A) = {\Lambda(\{a_1,a_2\})+\Lambda(\{a_3\})}\) and \(HL(\A_i) = \Lambda(\{a_i\})\). Since \(HL(\A_1)\wedge HL(\A_3) = \Lambda(\{a_1,a_3\}) \not \subseteq HL(\A)\), Proposition \ref{HLindep} implies that the HDAs \(\A_1\) and \(\A_3\) are not independent in \(\A\). For the same reason, \(\A_2\) and \(\A_3\) are not independent in \(\A\). The fact that \(HL(\A_1)\wedge HL(\A_2) = \Lambda(\{a_1,a_2\}) \subseteq HL(\A)\) suggests that \(\A_1\) and \(\A_2\) are independent in \(\A\). And indeed, \(\A_1\otimes \A_2\) is isomorphic in the required way to the sub-HDA \(\B\) of \(\A\) given by \((P_{\B})_0 = \{I_\A\}\), \((P_{\B})_1 = \{x_1,x_2\}\), \((P_\B)_2 = \{y\}\), and \(\Sigma_{\B} = \{a_1,a_2\}\).

	(ii) A small grocery store has three shopping baskets and two checkout counters. All customers behave the same:
	\begin{itemize}
		\item They wait until a basket is available and then start shopping.
		\item Once they have selected the products they wish to buy, they move to the checkouts and wait for their turn to pay. Since they are usually very polite and give others priority, the order in which they pay is unpredictable---even if there is a queue.
		\item Having paid, they return the basket and leave the store. 
		\item If they forgot to buy something, they repeat the procedure from the beginning.
	\end{itemize}  
	Focusing on the behavior of the customers with respect to the shared resources---the baskets and the checkouts---we may describe the shopping protocol as the program graph (in the sense of \cite{BaierKatoen}) depicted in Figure \ref{figex}. The actions modify two integer variables \(\mathtt{x}\) and \(\mathtt{y}\) counting the available shopping baskets and the free checkouts, respectively. They are defined as follows:
	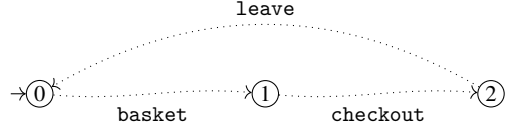
\begin{figure} 
		\center
		\begin{tikzpicture}[initial text={},on grid] 
		
		\node[state, initial by arrow, initial where=left, initial distance=0.2cm,minimum size=0pt,inner sep =1pt,fill=white] (q_0)  {\scalebox{0.85}{0}}; 
		
		\node[state,minimum size=0pt,inner sep =1pt,fill=white] (q_1) [right=of q_0,xshift=2cm] {\scalebox{0.85}{1}};

		\node[state,minimum size=0pt,inner sep =1pt,fill=white] (q_3) [right=of q_1,xshift=2cm] {\scalebox{0.85}{2}};

		\path[->] 
		(q_0) edge[below,dotted,inner sep =5pt,out=355,in=175] node {\scalebox{0.85}{\(\mathtt{basket}\)}} (q_1)
		(q_1) edge[below,dotted,inner sep =5pt,out=355,in=175] node {\scalebox{0.85}{\(\mathtt{checkout}\)}} (q_3)		
		(q_3) edge[above,dotted,out=150,in=30] node {\scalebox{0.85}{\(\mathtt{leave}\)}} (q_0);

		\end{tikzpicture}
		\caption{Program graph for Example \ref{secex} (ii)} \label{figex}
	\end{figure}
	\begin{itemize}
		\item \(\mathtt{basket}\): Wait until \(\mathtt{x>0}\), and then decrement \(\mathtt{x}\).
		\item \(\mathtt{checkout}\): Wait until \(\mathtt{y > 0}\), and then decrement \(\mathtt{y}\).
		\item \(\mathtt{leave}\): Increment both \(\mathtt{y}\) and \(\mathtt{x}\).
	\end{itemize}
	If we view \(\mathtt{x}\) and \(\mathtt{y}\) as semaphores and consider Dijkstra's \(\mathtt{P}\) and \(\mathtt{V}\) operations (see, \textit{e.g.}, \cite{DijkstraCSP}), then \(\mathtt{basket} = \mathtt{P(x)}\), \(\mathtt{checkout} = \mathtt{P(y)}\), and \(\mathtt{leave} = \mathtt{V(y); V(x)}\).
	
	Let us now consider a system of four customers executing the above protocol, and let us suppose that initially all shopping baskets and both checkout counters are free, \textit{i.e.}, \(\mathtt{x = 3}\) and \(\mathtt{y = 2}\). Assuming atomicity of the actions of the program graph, we may use the method of \cite{transhda}, implemented in the tool \(\mathtt{pg2hda}\) \cite{pg2hda}, to construct an HDA \(\A\) modeling the state space of the system. We do not need to know $\A$ in detail. Let us just mention that it is a 3-dimensional HDA with 563 cubes altogether and that its alphabet is the set 
	\[\Sigma_\A = \{\mathtt{basket_i}, \mathtt{checkout_i}, \mathtt{leave_i}\, |\, \mathtt{i} \in \{\mathtt{0},\mathtt{1},\mathtt{2},\mathtt{3}\}\}.\]
	The indexes of the labels are introduced to distinguish the four customers. 
	
	The homology language of \(\A\) with \(\Z_2\)-coefficients can be computed from \(\A\) with the aid of the software \texttt{CHomP} \cite{chomp}. It is clear that \(HL(\A)\) is generated by the unit of \(\Lambda(\Sigma_\A)\) in degree \(0\). In degree \(1\), \(HL(\A)\) is generated by the elements 
	\[\mathtt{basket_i} + \mathtt{checkout_i}+ \mathtt{leave_i},\quad \mathtt{i} \in \{\mathtt{0},\mathtt{1},\mathtt{2},\mathtt{3}\}.\]
	Each of these elements represents one of the customer processes executing alone. In degree \(2\), \(HL(\A)\) is generated by the products 
	\[(\mathtt{basket_i} + \mathtt{checkout_i}+ \mathtt{leave_i})\wedge (\mathtt{basket_j} + \mathtt{checkout_j}+ \mathtt{leave_j}),\quad \mathtt{i < j}.\]
	In view of Proposition \ref{HLindep}, this indicates that any two customers are independent and can proceed simultaneously without conflict if the other customers do nothing (or just talk). Since there are two checkouts, this is, of course, to be expected. Since there are no more than two checkouts, one would certainly also expect that no three customer processes are independent, despite the fact that there are three shopping baskets. And indeed, although \(\A\) has cubes of dimension 3, \(HL(\A)\) is trivial in degrees \(\geq 3\). We conclude that any two but no three customers are independent. Note that our analysis of the  independence structure of $\A$ has been carried out at the level of the homology language, without explicit mention of HDAs representing the customer processes. Note also that the homology language would have been the same for a store with only two baskets.
	
	As this example shows, the homology language does not necessarily uncover surprising features of concurrent systems. 
	Arguably, however, it encodes fundamental information on independence in HDAs. 
	
\end{exs}

\section{Weak equivalence} \label{secwe}

As pointed out in the introduction, weak equivalence is a coarse notion of equivalence for HDAs that focuses on a small number of fundamental features. Besides the trace language, the fundamental monoid, and the homology language, these are accessibility and coaccessibility. Weak equivalence is defined as the symmetric closure of a preorder called weak implementation. We show that both relations are compatible with the tensor product and, at least in the coaccessible case, the coproduct of HDAs. We also relate weak equivalence to the preorder of topological abstraction introduced in \cite{topabs} and adapt the results of that paper to provide conditions under which HDAs can be reduced to weakly equivalent smaller ones by collapsing and merging cubes.

\subsection{Accessible HDAs}

A state \(x\) of an HDA \(\A\) is called \emph{reachable} if there exists a path in \(\A\) from \(I_\A\) to \(x\). An HDA is called \emph{accessible} if all states are reachable. The proof of the following elementary fact is left to the reader: 

\begin{prop} \label{acc}
	Let \(\A\) and \(\B\) be two HDAs. The following statements are equivalent:
	\begin{enumerate}
		\item \(\A\) and \(\B\) are accessible.
		\item \(\A \otimes \B\) is accessible.
		\item \(\A + \B\) is accessible.
	\end{enumerate}
\end{prop}

\subsection{Coaccessible HDAs}

A state \(x\) is called \emph{coreachable} if there exists a path from \(x\) to a final state. An HDA is called \emph{coaccessible} if all states are coreachable. Coaccessibility guarantees the absence of very bad phenomena such as deadlocks. We omit the easy proof of the following proposition: 

\begin{prop} \label{coacc}
	Let \(\A\) and \(\B\) be two HDAs. If \(\A\) and \(\B\) are coaccessible, then so are \(\A \otimes \B\) and \(\A+ \B\). If \(\A \otimes \B\) is coaccessible, then so are \(\A\) and \(\B\).
\end{prop}

\begin{rem}
	Unfortunately, coaccessibility of \(\A+\B\) does not in general imply coaccessibility of \(\A\) and \(\B\). Indeed, consider two one-vertex HDAs \(\A\) and \(\B\), and suppose that \(F_\A = \{I_\A\}\) and \(F_\B = \emptyset\). Then \(\A\) and \(\A+\B\) are coaccessible but \(\B\) is not.	
\end{rem}

\subsection{Weak implementation and weak equivalence} \label{wedef}
We say that an HDA \(\A\) \emph{weakly implements} an HDA \(\B\) and write \({\A \sqsubseteq \B}\) if the following three conditions are satisfied: 
\begin{enumerate}
	\item If \(\B\) is accessible, then so is \(\A\). If \(\B\) is coaccessible, then so is \(\A\). 
	\item \((\Sigma_\A,D_\A) = (\Sigma_\B, D_\B)\), \(\pi(\A) \subseteq \pi(\B)\), and \(TL(\A) \subseteq TL(\B)\).
	\item \(HL(\A) \subseteq HL(\B)\).
\end{enumerate}
It is clear that weak implementation is a preorder on the class of HDAs. We say that two HDAs \(\A\) and \(\B\) are \emph{weakly equivalent} and write \(\A \simeq \B\) if \(\A \sqsubseteq \B\) and \(\B \sqsubseteq \A\). 

\begin{prop} \label{simulation}
	Let \(\A\) and \(\B\) be two HDAs over the same concurrent alphabet. If there exists a cubical dimap of HDAs \((f,id)\colon \A \to \B\) such that for all \(x \in (P_\A)_0\), \(x\) is reachable if \(f_0(x)\) is reachable and \(x\) is coreachable if \(f_0(x)\) is coreachable, then \(\A \sqsubseteq \B\). 
\end{prop}

\begin{proof}
	This follows from Propositions \ref{TLinc}, \ref{piinc}, and \ref{HLinc}.
\end{proof}

\begin{theor}
	Let \(\A\), \(\A'\), \(\B\), and \(\B'\) be HDAs such that \(\A \sqsubseteq \A'\) and \(\B \sqsubseteq \B'\). Then \(\A \otimes \B \sqsubseteq \A'\otimes \B'\). If \(\A\) and \(\B\) are coaccessible, then also \(\A + \B \sqsubseteq \A' + \B'\). 
\end{theor}

\begin{proof}
	This follows from Propositions \ref{TLtensor}, \ref{pitensor}, \ref{picoprod}, \ref{TLcoprod}, \ref{HLtensor}, \ref{HLcoprod}, \ref{acc}, and \ref{coacc}.  
\end{proof}

	\begin{cor}	
		Let \(\A\), \(\A'\), \(\B\), and \(\B'\) be HDAs such that \(\A \simeq \A'\) and \({\B \simeq \B'}\). Then \(\A \otimes \B \simeq \A'\otimes \B'\). If \(\A\), \(\A'\), \(\B\), and \(\B'\) are coaccessible, then also \({A + \B \simeq \A' + \B'}\).
	\end{cor}

\begin{rems} \label{bisimrem} 
	(i) Weak equivalence has been designed to be a coarse congruence for the tensor product and (as far as possible) the coproduct such that the trace language and the homology language are invariants. In certain situations, it might be convenient to modify the definition of weak equivalence. For instance, if the compatibility with the coproduct is not considered essential, the requirement on the fundamental monoid may be dropped. Another possible modification concerns accessibility. According to our definition, an HDA with unreachable states cannot be weakly equivalent to its accessible part. This is adequate if unreachable states are interpreted as representing problems such as dead code (see, \textit{e.g.}, \cite[p. 22]{FGHMR}). However, one might as well see unreachable states as just unreachable from the chosen initial state and prefer to define an equivalence where an HDA is always equivalent to its accessible part. To do so, one could define two HDAs to be equivalent if their accessible parts are weakly equivalent in the sense of this paper. For accessible HDAs, this concept of equivalence would coincide with the concept of weak equivalence proposed here.
	
	(ii) By Proposition \ref{simulation}, two HDAs over the same concurrent alphabet are weakly equivalent if there exist well-behaved cubical dimaps between them in both directions. Just as morphisms of HDAs, cubical dimaps may be seen as simulations, and so, from this point of view, weak equivalence is coarser than a kind of simulation equivalence. It should be pointed out in this context that HDAs that are history-preserving bisimilar in the sense of \cite{vanGlabbeek} need not be weakly equivalent. Consider, for example, an HDA \(\A\) with only one vertex and one edge, labeled \(a\). Then \(\A\) is history-preserving bisimilar to its unfolding \(\B\), which consists of an infinite sequence of edges, all labeled \(a\). On the other hand, since in degree 1, \(HL(\A)\) is generated by \(a\) but \(HL(\B) = 0\), \(\A\) and \(\B\) are not weakly equivalent. Consequently, these two HDAs are also not simulation equivalent in the above sense, and indeed, there is no cubical dimap from \(\A\) to \(\B\). If one wishes to define a concept of bisimilarity that is stronger than this notion of simulation equivalence, one possibility is to consider \(\mathsf{P}\)-bisimilarity in the sense of Joyal, Nielsen, and Winskel \cite{JoyalNielsenWinskel} where \(\mathsf{P}\) is the wide subcategory of the category of HDAs and cubical dimaps of the form \((f,id)\) whose morphisms are inclusions of sub-HDAs. It should be noted, though, that although it is not isomorphism, this concept of bisimilarity is very strong. 
\end{rems}

\subsection{Topological abstraction}

In \cite{topabs}, a preorder for HDAs has been introduced, called \emph{topological abstraction}. Roughly speaking, an HDA \(\A\) is a topological abstraction of an HDA \(\B\) if there exists a cubical dimap \(\A \to \B\) that is a homotopy equivalence inducing an isomorphism of trace categories and an isomorphism of homology graphs. The \emph{homology graph} of an HDA \(\A\) is the directed graph where the vertices are the homology classes of \(\A\) and there is an edge from a homology class \(\upsilon\) to a homology class \(\nu\) if there exist precubical subsets \(U, V \subseteq P_\A\) such that  \(\upsilon \in \im\, H_*(U \hookrightarrow P_\A)\), \(\nu \in \im\, H_*(V \hookrightarrow P_\A)\), and for all vertices \(u \in U_0\) and \(v \in V_0\), there exists a path from \(u\) to \(v\) \cite{hgraph}. The \emph{trace category} of an HDA \(\A\) is the category \(TC(\A)\) whose objects are the initial state, the final states, the minimal vertices (\textit{i.e.}, vertices without incoming edges), and the maximal vertices (\textit{i.e.}, vertices without outgoing edges) and whose morphisms are the dihomotopy classes of paths between these states \cite{weakmor}. A cubical dimap \((f,\sigma) \colon \A \to \B\) that preserves minimal and maximal vertices induces a functor \(f_*\colon TC(\A) \to TC(\B)\), which sends an object \(x\) to \(f_0(x)\) and a dihomotopy class \([\omega]\) to \([f^{\mathbb I}(\omega)]\). By the following proposition, topological abstraction is often stronger than weak equivalence:

\begin{prop}
	Let \(\A\) and \(\B\) be two accessible and coaccessible HDAs over the same concurrent alphabet, and let \((f, id) \colon \A \to \B\) be a cubical dimap that is a homotopy equivalence. Suppose that \(f\) preserves minimal and maximal vertices and that the functor \(f_*\colon TC(\A) \to TC(\B)\) is an isomorphism. Then \(\A \simeq \B\).
\end{prop}  

\begin{proof}
	By Propositions \ref{TLinc}, \ref{piinc}, and \ref{HLinc}, we only have to show that \(\pi(\B) \subseteq \pi(\A)\) and \(TL(\B) \subseteq TL(\A)\). Consider first a loop \(\beta \in P_\B^{\mathbb I}\) such that \(\beta(0) = I_\B\). Since \(f_0(I_\A) = I_\B\) and \(f_*\) is full, there exists a loop \(\alpha \in P_\A^{\mathbb I}\) such that \(\alpha(0) = I_\A\) and \(f_*([\alpha]) = [\beta]\). Hence \(f^{\mathbb I}(\alpha) \sim \beta\) and therefore \(\overline{\lambda}_\A (\alpha) = \overline{\lambda}_\B(f^{\mathbb I}(\alpha)) \equiv \overline{\lambda}_\B (\beta)\). Thus \([\overline{\lambda}_\B (\beta)] \in \pi(\A)\). 
	
	Consider now an element \(v\in TL(\B)\). Let \(\beta \in P_\B^{\mathbb I}\) be a path such that \(\beta (0) = I_\B\) and \(v \preceq [\overline{\lambda}_\B (\beta)]\). Since \(\B\) is coaccessible, we may suppose that \(\beta\) ends in a final state \(b\). Since \(f_*\) is full and surjective on objects, there exists a path \(\alpha \in P_\A^{\mathbb I}\) such that \(\alpha (0) = I_\A\) and \(f_*([\alpha]) = [\beta]\). As before, this implies that \(\overline{\lambda}_\A (\alpha) = \overline{\lambda}_\B(f^{\mathbb I}(\alpha)) \equiv \overline{\lambda}_\B (\beta)\). Hence \(v \preceq [\overline{\lambda}_\A (\alpha)]\) and therefore \(v \in TL(\A)\).
\end{proof}  

\begin{rem}
	An important difference between topological abstraction and weak equivalence is that weakly equivalent HDAs need not be homotopy equivalent and may have certain topological differences. Indeed, weak equivalence ignores zero-labeled homology classes (at least of dimension \(\geq 2\)). Examples of such classes include  torsion classes, differences of classes with the same label, classes that are noise (\textit{e.g.}, classes given by differences of cubes with the same boundary), and classes given by virtual boundaries (\textit{i.e.}, cycles that become boundaries in larger HDAs). In contexts where zero-labeled homology classes are essential, weak equivalence is too weak a notion of equivalence.   
\end{rem}

\subsection{Cube collapses} \label{collapse}
We shall now provide conditions under which collapsing a cube in an HDA yields a weakly equivalent HDA. We will consider elementary and vertex-star collapses. The definition of these concepts is based on the following construction: the \emph{star} of an element \(x\) of a precubical set \(P\) is the graded set \(\star(x)\) defined by
\[\star(x)_n = \{y \in P_n\,|\, x \in y_{\sharp}(\lbrbrak 0,1\rbrbrak ^{\otimes n})\}.\]
Thus \(\star(x)\) consists of \(x\) and all elements having \(x\) in their iterated boundary. The graded set \(P\setminus \star(x)\) is a precubical subset of \(P\). We say that a face \(d^k_ix\) of a regular cube \(x\) of \(P\) is \emph{free} if \(\star(d^k_ix) = \{x,d^k_ix\}\). In this case, the inclusion \(|P \setminus \star(d^k_ix)| \hookrightarrow |P|\) is a homotopy equivalence, and we say that \(P \setminus \star(d^k_ix)\) has been obtained from \(P\) through an \emph{elementary collapse} 
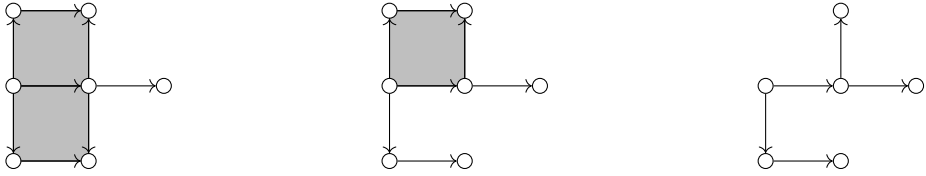
\begin{figure} 
	\center
	\begin{tikzpicture}[initial text={},on grid] 	
	\path[draw, fill=lightgray] (0,0)--(1,0)--(1,-1)--(0,-1)--cycle;     
	
	\path[draw, fill=lightgray] 
	
	(0,0)--(1,0)--(1,1)--(0,1)--cycle;

	\path[draw, fill=lightgray] 
	
	(5,0)--(6,0)--(6,1)--(5,1)--cycle;

	\node[state,minimum size=0pt,inner sep =2pt,fill=white] (p_0) at (0,0)  {}; 
	
	\node[state,minimum size=0pt,inner sep =2pt,fill=white] (p_6) [above=of p_0,xshift=0cm] {};
	
	\node[state,minimum size=0pt,inner sep =2pt,fill=white] (p_2) [right=of p_0,xshift=0cm] {};
	
	\node[state,minimum size=0pt,inner sep =2pt,fill=white] (p_1) [above=of p_2,xshift=0cm] {}; 
	
	\node[state,minimum size=0pt,inner sep =2pt,fill=white] (p_4) [right=of p_2,xshift=0cm] {};
	
	\node[state,minimum size=0pt,inner sep =2pt,fill=white] [below=of p_0, yshift=0cm] (p_3)   {};

	\node[state,minimum size=0pt,inner sep =2pt,fill=white] (p_5) [right=of p_3,xshift=0cm] {};

	\node[state,minimum size=0pt,inner sep =2pt,fill=white] (q_0) at (5,0)  {}; 
	
	\node[state,minimum size=0pt,inner sep =2pt,fill=white] (q_6) [above=of q_0,xshift=0cm] {};
	
	\node[state,minimum size=0pt,inner sep =2pt,fill=white] (q_2) [right=of q_0,xshift=0cm] {};
	
	\node[state,minimum size=0pt,inner sep =2pt,fill=white] (q_1) [above=of q_2,xshift=0cm] {};
	
	\node[state,minimum size=0pt,inner sep =2pt,fill=white] (q_4) [right=of q_2,xshift=0cm] {};
	
	\node[state,minimum size=0pt,inner sep =2pt,fill=white] [below=of q_0, yshift=0cm] (q_3)   {};

	\node[state,minimum size=0pt,inner sep =2pt,fill=white] (q_5) [right=of q_3,xshift=0cm] {};

	\node[state,minimum size=0pt,inner sep =2pt,fill=white] (r_0) at (10,0)  {};

	\node[state,minimum size=0pt,inner sep =2pt,fill=white] (r_2) [right=of r_0,xshift=0cm] {};
	
	\node[state,minimum size=0pt,inner sep =2pt,fill=white] (r_1) [above=of r_2,xshift=0cm] {};
	
	\node[state,minimum size=0pt,inner sep =2pt,fill=white] (r_4) [right=of r_2,xshift=0cm] {};
	
	\node[state,minimum size=0pt,inner sep =2pt,fill=white] [below=of r_0, yshift=0cm] (r_3)   {};

	\node[state,minimum size=0pt,inner sep =2pt,fill=white] (r_5) [right=of r_3,xshift=0cm] {};

	\path[->] 
	(p_2) edge[above] node {} (p_1)
	(p_6) edge[above] node {} (p_1)
	(p_0) edge[above] node {} (p_2)
	(p_0) edge[above] node {} (p_6)
	(p_3) edge[below]  node {} (p_5)
	(p_0) edge[left]  node {} (p_3)
	(p_2) edge[right]  node {} (p_5)
	(p_2) edge[right]  node {} (p_4);

	\path[->] 
	(q_2) edge[above] node {} (q_1)
	(q_6) edge[above] node {} (q_1)
	(q_0) edge[above] node {} (q_2)
	(q_0) edge[above] node {} (q_6)
	(q_3) edge[below]  node {} (q_5)
	(q_0) edge[left]  node {} (q_3)
	(q_2) edge[left]  node {} (q_4)
	;

	\path[->] 
	(r_2) edge[above] node {} (r_1)
	(r_0) edge[above] node {} (r_2)
	(r_3) edge[below]  node {} (r_5)
	(r_0) edge[left]  node {} (r_3)
	(r_2) edge[left]  node {} (r_4)
	;
	
	\end{tikzpicture}
	\caption{An elementary collapse followed by a vertex-star collapse}\label{ecollapse}
\end{figure}
(see Figure \ref{ecollapse} for a picture). If \(x\in P\) is a regular cube of degree \(n \geq 2\) and \(k_1, \dots, k_n \in\{0,1\}\) are indexes such that at least one \(k_i = 0\), at least one $k_i = 1$, and \(\star(d^{k_n}_1\cdots d^{k_1}_1x) \subseteq x_{\sharp}(\lbrbrak 0,1 \rbrbrak ^{\otimes n})\), then the inclusion   \(|P \setminus \star(d^{k_n}_1\cdots d^{k_1}_1x)| \hookrightarrow |P|\) is a homotopy equivalence and we say that \({P \setminus \star(d^{k_n}_1\cdots d^{k_1}_1x)}\) has been obtained from \(P\) through a \emph{vertex-star collapse}
(see Figure \ref{ecollapse}).

In degrees \(\geq 3\), elementary collapses always yield weakly equivalent HDAs:

\begin{prop}
	Let \(\A\) be an HDA, and let \(x\) be a regular cube of degree \(n \geq 3\) with free face \(d^k_ix\). Consider the sub-HDA \(\B\subseteq \A\) defined by \(P_\B = P_\A \setminus \star(d^k_ix)\) and \(\Sigma_\B = \Sigma_\A\). Then \(\A \simeq \B\).
\end{prop}

\begin{proof}
	This follows from Proposition \ref{HLinc} and the fact that \(\A\) and \(\B\) agree in degrees \(\leq 1\). 
\end{proof}

Elementary 2-cube collapses are more delicate. We first deal with the case where the free face is a back face:

\begin{theor} \label{TL2}
	Let \(\A\) be an HDA, and let \(x\) be a regular \(2\)-cube with free face \(d^1_ix\) \((i \in \{1,2\})\). Consider the sub-HDA \(\B \subseteq \A\) defined by \(P_\B = P_\A \setminus \star(d^1_ix)\) and \(\Sigma_\B = \Sigma_\A\). Suppose that there exists an edge \(y\) in \(\B\) such that \(d_1^{0}y = d_1^{0}d_i^1x\), and suppose that for every path \(\omega \in P_\B^ {\mathbb I}\) with \(\omega (\length_\omega) = d^{0}_1d^1_ix\) and \({\omega (0) \in \{I_\A, d^1_1d^ 1_1x, d_1^1y\}}\), there exists a path \(\nu \in P_\B^{\mathbb I}\) such that \(\nu (0) = \omega(0)\), \(\nu (\length_\nu) = d^0_1d^0_1x\), and \(\nu \cdot (d^0_{3-i}x)_\sharp \sim \omega\) in \(\A\). Then \(\A \simeq \B\). 
\end{theor}

\begin{proof}
	We adapt arguments given in the proofs of \cite[Lemma 4.4.3, Thm. 6.3.1]{topabs}. Since the inclusion \(|P_\B| \hookrightarrow |P_\A|\) is a homotopy equivalence, \(HL(\A) = HL(\B)\). By Propositions \ref{TLinc} and \ref{piinc}, we have \(TL(\B) \subseteq TL(\A)\) and \(\pi(\B) \subseteq \pi(\A)\). For the reverse inclusions, it suffices to show that every path \(\omega \in P_\A^{\mathbb I}\) with \({\omega (0)} = I_\A\) is dihomotopic to a path  \(\omega'\in P_\B^{\mathbb I}\). So consider \(\omega \in P_\A^{\mathbb I}\) with \({\omega (0)} = I_\A\). We may suppose that \(\omega \notin P_\B^{\mathbb I}\). Write \(\omega\) as a concatenation \[\omega = \omega_0\cdot (d^1_ix)_{\sharp}\cdot  \omega_1 \cdots \omega_{r-1}\cdot (d^1_ix)_{\sharp} \cdot \omega_{r}\]
	where each \(\omega_j\) is a path in \(P_\B\). By our assumptions, there exist paths \(\bar \omega_j \in P_\B^ {\mathbb I}\) \((0\leq j < r)\) such that \(\bar \omega_j(0) = \omega_j(0)\), \(\bar \omega_j(\length_{\bar \omega_j}) = d_1^0d_1^0x\), and \(\bar \omega_j\cdot (d^0_{3-i}x)_\sharp \sim \omega_j\). 
	Set
	\[\omega' = \bar \omega_0\cdot (d^0_{i}x)_{\sharp}\cdot (d^1_{3-i}x)_{\sharp}\cdot  \bar \omega_1 \cdots \bar \omega_{r-1}\cdot (d^0_{i}x)_{\sharp}\cdot (d^1_{3-i}x)_{\sharp} \cdot \omega_{r}.\]
	Since \(x\) is regular, \(\omega'\in P_\B^{\mathbb I}\). It is clear that \(\omega' \sim \omega\). 
	
	Since \(\A\) and \(\B\) have the same vertices and the same initial and final states and every path in \(\B\) is a path in \(\A\), it is clear that \(\A\) is (co)accessible  if \(\B\) is (co)accessible. Suppose that \(\A\) is accessible, and consider a vertex \(b \in (P_\B)_0\). Then there exists a path \(\omega \in P_\A^{\mathbb I}\) from \(I_\A\) to \(b\). As shown above, there exists a path  \(\omega'\in P_\B^{\mathbb I}\) such that \(\omega' \sim \omega\). It follows that \(\B\) is accessible. Suppose that \(\A\) is coaccessible. By our hypothesis, there exists an edge \(y \not=d^1_ix\) starting in \(d^0_1d^1_ix\) such that for every path \(\omega \in P_\B^ {\mathbb I}\) with \(\omega (\length_\omega) = d^{0}_1d^1_ix\) and \({\omega (0) = d_1^1y}\), there exists a path \(\nu \in P_\B^{\mathbb I}\) such that \(\nu (0) = d_1^1y\), \(\nu (\length_\nu) = d^0_1d^0_1x\), and \(\nu \cdot (d^0_{3-i}x)_\sharp \sim \omega\). We show first that \(d^1_1y\) is coreachable in \(\B\). Since \(\A\) is coaccessible, there exists a path \(\gamma \in P_\A^{\mathbb I}\) from \(d^1_1y\) to a vertex \(c \in F_\A = F_\B\). If \(\gamma \in P_\B^{\mathbb I}\), we have  nothing to show. If \(\gamma \notin P_\B^{\mathbb I}\), it begins with a path from \(d^1_1y\) to \(d^0_1d^1_ix\). Let \(\omega\) be a shortest such path. Then \(\omega\in P_\B^{\mathbb I}\), and so we may choose a path \(\nu \in P_\B^{\mathbb I}\) such that \(\nu (0) = d_1^1y\), \(\nu (\length_\nu) = d^0_1d^0_1x\), and \(\nu\cdot (d^0_{3-i}x)_\sharp \sim  \omega\). Since \(\gamma \notin P_\B^{\mathbb I}\), \(\gamma\) terminates with a path from \(d^1_1d^1_1x\) to \(c\). Let \(\beta \) be a shortest such path. Then \(\beta\in P_\B^{\mathbb I}\). The concatenation \(\nu \cdot (d^0_ix)_{\sharp}\cdot (d^1_{3-i}x)_{\sharp}\cdot \beta\) is a path in \(\B\) from \(d^1_1y\) to \(c\). Hence \(d^1_1y\) is coreachable in \(\B\). Consider now an arbitrary vertex \(b \in (P_\B)_0\). Then there exists a path \(\alpha \in P_\A^{\mathbb I}\) from \(b\) to a final state. If \(\alpha \in P_\B^{\mathbb I}\), \(b\) is coreachable in \(\B\). If \(\alpha \notin P_\B^{\mathbb I}\), it begins with a path in \(\B\) from \(b\) to \(d^0_1d^1_ix = d^0_1y\). Hence there exists a path in \(\B\) from \(b\) to \(d^1_1y\). Since \(d^1_1y\) is coreachable in \(\B\), it follows that \(b\) is coreachable in \(\B\). Hence \(\B\) is coaccessible if \(\A\) is coaccessible.	
\end{proof}

For elementary collapses of 2-cubes with a free front face, we state the following fact, which is proved by adapting the arguments given in the proof of Theorem \ref{TL2}:

\begin{theor} \label{TL2b}
	Let \(\A\) be an HDA, and let \(x\) be a regular \(2\)-cube with free face \(d^0_ix\) \((i \in \{1,2\})\). Suppose that \(d^1_1d^0_ix \notin \{I_\A\}\cup F_\A\) and that there exists an edge \(y \not= d^0_ix\) such that \(d_1^{1}y = d_1^{1}d_i^0x\). Suppose also that there is no edge \(z \not= d^1_{3-i}x\) such that \(d_1^0z = d_1^1d^0_ix\). Consider the sub-HDA \(\B\subseteq \A\) defined by \(P_\B = P_\A \setminus \star(d^0_ix)\) and \(\Sigma_\B = \Sigma_\A\). Then \(\A \simeq \B\). 
\end{theor}

Regarding vertex-star collapses, we have the following result:

\begin{theor} \label{TLvertex}
	Let \(x\) be a regular cube of degree \(n\geq 2\) of an HDA \(\A\), and let \(k_1, \dots, k_n \in\{0,1\}\) such at least one \(k_i = 0\), at least one \(k_i = 1\), \(d^{k_n}_1\cdots d^{k_1}_1x \notin \{I_\A\} \cup F_\A\), and \(\star(d^{k_n}_1\cdots d^{k_1}_1x) \subseteq x_{\sharp}(\lbrbrak 0,1 \rbrbrak ^{\otimes n})\). Consider the sub-HDA \(\B\subseteq \A\) defined by \(P_\B = {P_\A \setminus \star(d^{k_n}_1\cdots d^{k_1}_1x)}\) and \(\Sigma_\B = \Sigma_\A\). Then \(\A \simeq \B\).	
\end{theor}

\begin{proof}
	Since there exists a path in \(\A\) from \(d^{0}_1\cdots d^{0}_1x\) to \(d^{k_n}_1\cdots d^{k_1}_1x\), \(\A\) is accessible if \(\B\) is accessible. Since there exists a path in \(\A\) from  \(d^{k_n}_1\cdots d^{k_1}_1x\) to \(d^{1}_1\cdots d^{1}_1x\), \(\A\) is coaccessible if \(\B\) is coaccessible. By \cite[Thm. 4.6.1]{topabs}, \cite{Misamore}, every path \(\omega \in P_\A^{\mathbb I}\) with endpoints in \(\B\) is dihomotopic to a path \(\omega' \in P_\B^{\mathbb I}\). Since \(\A\) and \(\B\) have the same initial and final states, this implies that \(\B\) is (co)accessible if \(\A\) is (co)accessible. By Propositions \ref{TLinc} and \ref{piinc}, it also follows that \(TL(\A) = TL(\B)\) and \(\pi(\A) = \pi(\B)\). Since the inclusion \(|P_\B| \hookrightarrow |P_\A|\) is a homotopy equivalence, \(HL(\A) = HL(\B)\).   
\end{proof}

\subsection{Cube merging} \label{merging}

Let \(\A\) and \(\B\) be two HDAs over the same concurrent alphabet. If \(\B\) is weakly regular and \(\A\) is obtained from \(\B\) by merging cubes by means of a subdivision homeomorphism (see Section \ref{cubdi}), then \(\A\) and \(\B\) are weakly equivalent. More precisely, we have the following theorem:

\begin{theor} \label{merge}
	Let \((f, id) \colon \A \to \B\) be an elementary cubical dimap such that \(f\colon |P_\A| \to |P_\B|\) is a subdivision homeomorphism and \(f_0(F_\A) = F_\B\). If \(\B\) is weakly regular, then \(\A \simeq \B\).
\end{theor}

\begin{proof}
	The fact that \(\A\) is (co)accessible if and only if \(\B\) is (co)accessible is shown in \cite[Thm. 6.2.3]{topabs}. Since \(f\) is a homotopy equivalence, \(HL(\A) = HL(\B)\). By Propositions \ref{TLinc} and \ref{piinc}, it remains to prove that \(\pi(\B) \subseteq \pi(\A)\) and \(TL(\B) \subseteq TL(\B)\). By \cite[Prop 3.4.1]{topabs}, \(\A\) is weakly regular. By \cite[Prop. 4.7.4]{weakmor}, it follows that for every path \(\beta \in P_\B^{\mathbb I}\) from \(I_\B\) to a vertex of the form \(f_0(a)\) where \(a \in (P_\A)_0\), there exists a path \(\alpha \in P_\A^{\mathbb I}\) from \(I_\A\) to \(a\) such that \(f^{\mathbb I}(\alpha) \sim \beta\). This immediately implies that \(\pi(\B) \subseteq \pi(\A)\). Consider \(m \in TL(\B)\). Let \(\omega \in P_\B^{\mathbb I}\) be a path starting in \(I_\B\) such that \(m \preceq [\overline{\lambda}_\B(\omega)]\). Consider the vertex \(b = \omega(\length_\omega)\). Then there exist an integer \(n \geq 0\) and an element \(c(b) \in (P_\A)_n\), called the \emph{carrier} of \(b\), such that \(f([c(b),u]) = [b,()]\) for some \(u \in \mathopen] 0,1\mathclose[^{n}\). Since \(f\) is a subdivision homeomorphism, there exist integers \(l_1, \dots, l_n\geq 1\), a morphism of precubical sets \(\chi \colon \bigotimes_{i=1}^n \lbrbrak 0,l_i\rbrbrak \to P_\B\), and increasing homeomorphisms \(\phi_i \colon [0,1] \to  [0,l_i ]\) \((i \in \{1, \dots, n\})\) such that the following diagram commutes: 
		\[\xymatrix{
			|\lbrbrak	0, 1\rbrbrak^{\otimes n}| \ar[r]^(0.55){\approx} \ar[d]_{|c(b)_{\sharp}|} & [0,1]^n \ar[rr]^{\phi_1 \times \cdots \times \phi_n} && \prod \limits_{i=1}^n [0,l_i]  \ar[r]^(0.45){\approx} & |\bigotimes \limits_{i=1}^n \lbrbrak 0,l_i\rbrbrak | \ar[d]^{|\chi|}\\
			|P_\A| \ar[rrrr]_{f} &&&& |P_\B|
		}\]
		Let us write \(\phi\) to denote the upper horizontal composite. Let \(r \in \N\), \(x \in (\bigotimes \limits_{i=1}^n \lbrbrak 0,l_i\rbrbrak)_r\), and \(v \in \mathopen]0,1\mathclose[^{r}\) be the uniquely determined elements such that \(\phi([\iota_n,u]) = [x,v]\). Since \([\chi(x),v] = f([c(b),u]) = [b,()]\), \(x\) is a vertex and \(\chi(x) = b\). Let \(\nu\) be a path in \(\bigotimes \limits_{i=1}^n \lbrbrak 0,l_i\rbrbrak\) from \(x\) to \((l_1, \dots, l_n)\). Since the \(\phi_i\) are increasing homeomorphisms, \(\phi([(1, \dots, 1),()]) = [(l_1, \dots , l_n),()]\). Hence  \([\chi(l_1,\dots,l_n),()] = f([d^1_1\cdots d^1_1c(b),()])\) and therefore \(\chi(l_1, \dots, l_n) = f_0(d^1_1\cdots d^1_1c(b))\). Thus \(\beta = \omega\cdot (\chi \circ \nu)\) is a path from \(I_\B\) to \(f_0(d^1_1\cdots d^1_1c(b))\). It follows that there exists a path \(\alpha \in P_\A^{\mathbb I}\) from \(I_\A\) to \(d^1_1\cdots d^1_1c(b)\) such that \(f^{\mathbb I}(\alpha) \sim \beta\). Since \(m \preceq [\overline{\lambda}_\B(\omega)]\), also \(m \preceq [\overline{\lambda}_\B(\beta)] = [\overline{\lambda}_\A(\alpha)]\). It follows that \(TL(\B) \subseteq TL(\A)\). 
\end{proof}

\subsection{Example} 

\begin{figure}[t]
	\center
	\begin{tikzpicture}[initial text={},on grid]

	\path[draw, fill=lightgray] (0,0)--(1,0)--(1,-1)--(0,-1)--cycle;    
	
	\path[draw, fill=lightgray] (1,0)--(2,0)--(2,-1)--(1,-1)--cycle;
	
	\path[draw, fill=lightgray] (2,0)--(4,0)--(3,-1)--(2,-1)--cycle;
	
	\path[draw, fill=lightgray] (0,-1)--(1,-1)--(1,-3)--(0,-3)--cycle;
	
	\path[draw, fill=lightgray] (1,-1)--(2,-1)--(2,-2)--(1,-3)--cycle;
	
	\path[draw, fill=lightgray] (3,-1)--(4,0)--(4,-2)--(3,-2)--cycle;
	
	\path[draw, fill=lightgray] (4,0)--(5,0)--(5,-2)--(4,-2)--cycle;
	
	\path[draw, fill=lightgray] (2,-2)--(3,-2)--(3,-3)--(1,-3)--cycle;
	
	\path[draw, fill=lightgray] (3,-2)--(4,-2)--(4,-3)--(3,-3)--cycle;
	
	\path[draw, fill=lightgray] (4,-2)--(5,-2)--(5,-3)--(4,-3)--cycle;
	
	\draw[] (3,-1) to [out=225,in=45] (-1.35,-1.65); 
	
	\draw[->] (-1.35,-1.65) to [out=225,in=180] (-0.1,-4);
	
	\node at (-2,-3) {\scalebox{0.85}{$\mathtt{t\!\coloneqq_0\! 1}$}};
	
	\draw[] (2,-2) to [out=45,in=225] (6.35,-1.35); 
	
	\draw[->] (6.35,-1.35) to [out=45,in=0] (5.1,1);

	\node at (7,0) {\scalebox{0.85}{$\mathtt{t\!\coloneqq_1\! 0}$}};

	\node[state,minimum size=0pt,inner sep =2pt,fill=white] (q_0) at (0,0)  {}; 
	
	\node[state,minimum size=0pt,inner sep =2pt,fill=white,accepting, initial,initial where=above,initial distance=0.2cm] (q_1) [right=of q_0,xshift=0cm] {};
	
	\node[state,minimum size=0pt,inner sep =2pt,fill=white] (q_2) [right=of q_1,xshift=0cm] {};
	
	\node[state,minimum size=0pt,inner sep =2pt,fill=white] (q_3) [right=of q_2,xshift=1cm] {};
	
	\node[state,minimum size=0pt,inner sep =2pt,fill=white] (q_4) [right=of q_3,xshift=0cm] {};
	
	\node[state,minimum size=0pt,inner sep =2pt,fill=white] (q_5) [below=of q_0,xshift=0cm] {};  
	
	\node[state,minimum size=0pt,inner sep =2pt,fill=white] (q_6) [right=of q_5,xshift=0cm] {};
	
	\node[state,minimum size=0pt,inner sep =2pt,fill=white] (q_7) [right=of q_6,xshift=0cm] {};
	
	\node[state,minimum size=0pt,inner sep =2pt,fill=white] (q_8) [right=of q_7,xshift=0cm] {};
	
	\node[state,minimum size=0pt,inner sep =2pt,fill=white] (q_9) [below=of q_7,xshift=0cm] {};
	
	\node[state,minimum size=0pt,inner sep =2pt,fill=white] (q_10) [right=of q_9,xshift=0cm] {};
	
	\node[state,minimum size=0pt,inner sep =2pt,fill=white] (q_11) [right=of q_10,xshift=0cm] {};
	
	\node[state,minimum size=0pt,inner sep =2pt,fill=white] (q_12) [right=of q_11,xshift=0cm] {};
	
	\node[state,minimum size=0pt,inner sep =2pt,fill=white] (q_17) [below=of q_12,xshift=0cm] {};
	
	\node[state,minimum size=0pt,inner sep =2pt,fill=white,accepting] (q_16) [left=of q_17,xshift=0cm] {};
	
	\node[state,minimum size=0pt,inner sep =2pt,fill=white] (q_15) [left=of q_16,xshift=0cm] {};
	
	\node[state,minimum size=0pt,inner sep =2pt,fill=white] (q_14) [left=of q_15,xshift=-1cm] {};
	
	\node[state,minimum size=0pt,inner sep =2pt,fill=white] (q_13) [left=of q_14,xshift=0cm] {};
	
	\node[state,minimum size=0pt,inner sep =2pt,fill=white] (p_0) [above=of q_4,xshift=0cm] {};
	
	\node[state,minimum size=0pt,inner sep =2pt,fill=white] (p_1) [below=of q_13,xshift=0cm] {};

	\path[->] 
	(p_0) edge[right] node[] {\scalebox{0.85}{$\mathtt{crit_0}$}} (q_4)
	
	(q_0) edge[below] node[] {\scalebox{0.85}{$\mathtt{b_1\!\!\coloneqq_1\!\! 0}$}} (q_1)
	(q_1) edge[below] node[] {\scalebox{0.85}{$\mathtt{b_1\!\!\coloneqq_1\!\! 1}$}} (q_2)
	(q_2) edge[below] node[] {\scalebox{0.85}{$\mathtt{t\!\!\coloneqq_1\!\! 0}$}} (q_3)
	(q_4) edge[above] node[] {\scalebox{0.85}{$\mathtt{b_0\!\!\coloneqq_0\!\! 0}$}} (q_3)
	
	(q_0) edge[left,] node[] {\scalebox{0.85}{$\mathtt{b_0\!\!\coloneqq_0\!\! 1}$}} (q_5)
	(q_1) edge[above] node {} (q_6)
	(q_2) edge[above] node {} (q_7)
	(q_3) edge[above] node {} (q_8)
	
	(q_5) edge[above] node {} (q_6)
	(q_6) edge[above] node {} (q_7)
	(q_7) edge[above] node {} (q_8)
	
	(q_7) edge[above] node {} (q_9)
	(q_10) edge[above] node {} (q_8)
	(q_11) edge[above] node {} (q_3)
	(q_12) edge[right] node[] {\scalebox{0.85}{$\mathtt{t\!\!\coloneqq_1\!\! 0}$}} (q_4)
	
	(q_10) edge[above] node {} (q_9)
	(q_11) edge[above] node {} (q_10)
	(q_12) edge[above] node {} (q_11)
	
	(q_5) edge[left] node[] {\scalebox{0.85}{$\mathtt{t\!\!\coloneqq_0\!\! 1}$}} (q_13)
	(q_6) edge[above] node {} (q_14)
	(q_14) edge[above] node {} (q_9)
	(q_15) edge[above] node {} (q_10)
	(q_16) edge[above] node {} (q_11)
	(q_17) edge[right] node[] {\scalebox{0.85}{$\mathtt{b_1\!\!\coloneqq_1\!\! 1}$}} (q_12)
	
	(q_13) edge[below] node[] {\scalebox{0.85}{$\mathtt{b_1\!\!\coloneqq_1\!\! 0}$}} (q_14)
	(q_15) edge[above] node[] {\scalebox{0.85}{$\mathtt{t\!\!\coloneqq_0\!\! 1}$}} (q_14)
	(q_16) edge[above] node[] {\scalebox{0.85}{$\mathtt{b_0\!\!\coloneqq_0\!\! 1}$}} (q_15)
	(q_17) edge[above] node[] {\scalebox{0.85}{$\mathtt{b_0\!\!\coloneqq_0\!\! 0}$}} (q_16)
	
	(p_1) edge[left] node[] {\scalebox{0.85}{$\mathtt{crit_1}$}} (q_13)
	
	(q_3) edge[above, bend right] node[] {\scalebox{0.85}{$\mathtt{crit_1}$}} (q_0)
	(q_14) edge[below, bend right] node[] {\scalebox{0.85}{$\mathtt{crit_0}$}} (q_17)
	;

	\end{tikzpicture}
	\caption{HDA for Peterson's algorithm (parallel arrows have the same label)}\label{peterHDA}
\end{figure}
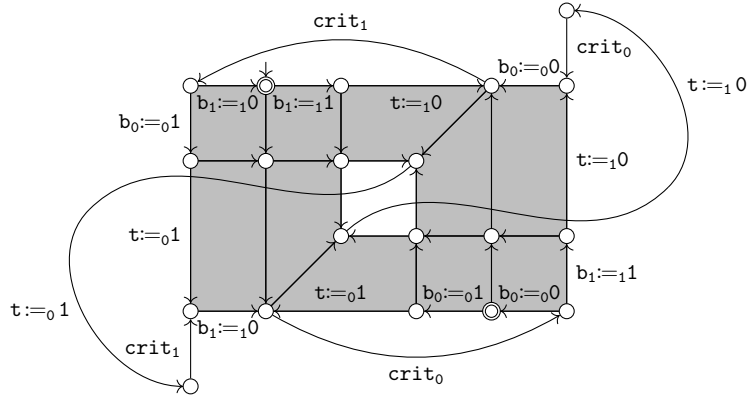

An HDA modeling the accessible part of the system given by
Peterson's mutual exclusion algorithm \cite{Peterson} is depicted in Figure \ref{peterHDA}. The concurrent alphabet is the pair \((\Sigma, D)\) where \(\Sigma\) is the set of edge labels and \(D\) is the canonical dependence relation (see Section \ref{HDAdef}). Peterson's algorithm is based on three shared variables---namely, the boolean variables \(\mathtt{b_0}\) and \(\mathtt{b_1}\) and the turn variable \(\mathtt{t}\), whose possible values are the process IDs, say \(\mathtt{0}\) and \(\mathtt{1}\). Process \(\mathtt{i}\)  executes the following protocol: 
\begin{itemize}
	\item Set \(\mathtt{b_i}\) to \(\mathtt{1}\) to indicate the intention to enter the critical section.
	\item Set \(\mathtt{t}\) to \(\mathtt{1-i}\) to give priority to the other process.
	\item Wait until \(\mathtt{b_{1-i} = 0}\) or \(\mathtt{t = i}\), and then enter the critical section. 
	\item Leave the critical section setting \(\mathtt{b_i}\) to \(\mathtt{0}\).
	\item Repeat the procedure from the beginning. 
\end{itemize}	
As explained in more detail in \cite{topabs}, the HDA for Peterson's algorithm may be reduced to the one depicted in Figure \ref{peterHDAmin} by collapsing and merging cubes in the way discussed in Sections 6.5 and 6.6 and, more precisely, using Theorems \ref{TL2}, \ref{TL2b}, \ref{TLvertex}, and \ref{merge}. Consequently, the two HDAs are weakly equivalent. 

\begin{figure}[t]
	\center
	\begin{tikzpicture}[initial text={},on grid]

	\draw[] (4,0) to [out=225,in=45] (-0.35,-1.65); 
	
	\draw[->] (-0.35,-1.65) to [out=225,in=180] (0.9,-3);
	
	\node[align=left] at (-1,-2) {\scalebox{0.85}{$\mathtt{b_0\!\!\coloneqq_0\!\! 1;}$}\\\scalebox{0.85}{$\mathtt{t\!\!\coloneqq_0\!\! 1;}$}\\\scalebox{0.85}{$\mathtt{crit_1;}$}\\\scalebox{0.85}{$\mathtt{b_1\!\!\coloneqq_1\!\! 0}$}};

	\draw[] (1,-3) to [out=45,in=225] (5.35,-1.35); 
	
	\draw[->] (5.35,-1.35) to [out=45,in=0] (4.1,0);
	
	\node[align=left] at (6.1,-1) {\scalebox{0.85}{$\mathtt{b_1\!\!\coloneqq_1\!\! 1;}$}\\\scalebox{0.85}{$\mathtt{t\!\!\coloneqq_1\!\! 0;}$}\\\scalebox{0.85}{$\mathtt{crit_0;}$}\\\scalebox{0.85}{$\mathtt{b_0\!\!\coloneqq_0\!\! 0}$}};
	
	\node[align=left] at (1.6,-1.8) {\scalebox{0.85}{$\mathtt{b_0\!\!\coloneqq_0\!\! 1;}$}\\\scalebox{0.85}{$\mathtt{t\!\!\coloneqq_0\!\! 1}$}}; 
	
	\node[align=left] at (3.4,-1.3) {\scalebox{0.85}{$\mathtt{b_1\!\!\coloneqq_1\!\! 1;}$}\\\scalebox{0.85}{$\mathtt{t\!\!\coloneqq_1\!\! 0}$}};

	\node[state,minimum size=0pt,inner sep =2pt,fill=white,accepting,initial,initial where=above,initial distance=0.2cm] (q_1) at (1,0) {};

	\node[state,minimum size=0pt,inner sep =2pt,fill=white] (q_3) [right=of q_2,xshift=1cm] {};

	\node[state,minimum size=0pt,inner sep =2pt,fill=white,accepting] (q_16) [left=of q_17,xshift=0cm] {};

	\node[state,minimum size=0pt,inner sep =2pt,fill=white] (q_14) [left=of q_15,xshift=-1cm] {};

	\path[->] 
	
	(q_1) edge[below] node {\scalebox{0.85}{$\mathtt{b_1\!\!\coloneqq_1\!\! 1;t\!\!\coloneqq_1\!\! 0}$}} (q_3)

	(q_1) edge[right] node{}(q_14)

	(q_16) edge[right] node {} (q_3)

	(q_16) edge[above] node {\scalebox{0.85}{$\mathtt{b_0\!\!\coloneqq_0\!\! 1;t\!\!\coloneqq_0\!\! 1}$}} (q_14)

	(q_3) edge[above, bend right] node {\scalebox{0.85}{$\mathtt{crit_1;b_1\!\!\coloneqq_1\!\! 0}$}} (q_1)
	(q_14) edge[below, bend right] node {\scalebox{0.85}{$\mathtt{crit_0;b_0\!\!\coloneqq_0\!\! 0}$}} (q_16)
	;
	\end{tikzpicture}
	\caption{Reduced HDA for Peterson's algorithm}\label{peterHDAmin} 
\end{figure}
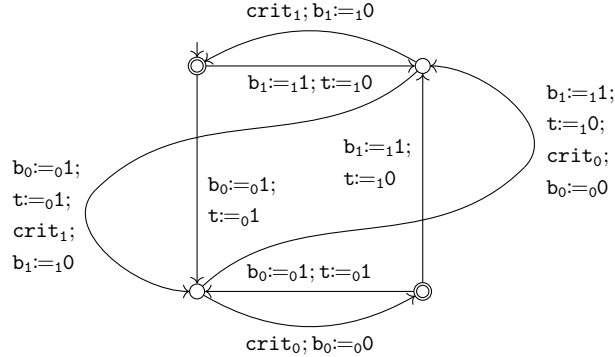

Since weakly equivalent HDAs have the same trace language, they have the same saturated safety properties (see Proposition  \ref{TLsatsafe}). In the case of Peteron's algorithm, such properties may thus be verified for the small HDA in Figure \ref{peterHDAmin} instead of for the bigger one depicted in Figure \ref{peterHDA}. This applies in particular to mutual exclusion, which is the saturated safety property given by the set 
\[\bigcup _{\mathtt{i=0}}^\mathtt{1}\Sigma ^* \cdot \{\mathtt{crit_i}\}\cdot (\Sigma\setminus \{\mathtt{b_i\!\!\coloneqq_i\!\!0}\})^ * \cdot \{\mathtt{crit_{1-i}}\}.\]
Another important feature of Peterson's algorithm is starvation freedom: a process that requests access to the critical section will eventually obtain it. This liveness property can be inferred from the saturated safety properties given by the sets 
\[\Sigma^* \cdot \{a\} \cdot (\Sigma\setminus \{\mathtt{crit_0},\mathtt{crit_1}\})^*\cdot \{a\} \quad (a \in \Sigma \setminus \{\mathtt{crit_0},\mathtt{crit_1}	\})\]
and 
\[\Sigma ^* \cdot \{\mathtt{b_i\!\!\coloneqq_i\!\!1}\} \cdot (\Sigma \setminus \{\mathtt{crit_i}\})^*\cdot \{\mathtt{crit_{1-i}}\} \cdot (\Sigma\setminus \{\mathtt{crit_i}\})^*\cdot \{\mathtt{crit_{1-i}}\} \quad (\mathtt{i} \in \{\mathtt{0},\mathtt{1}\}).\]
Starvation freedom of Peterson's algorithm can thus be established using any HDA weakly equivalent to the one of Figure \ref{peterHDA}. It should be noted, however, that the trace language only contains information on saturated safety properties (see Proposition  \ref{TLsatsafe}) and that therefore weak equivalence does not preserve liveness properties in general.

\section{Concluding remarks}
This paper introduced weak equivalence, a coarse notion of equivalence for higher-dimensional automata. Although equivalences for HDAs do not really fit into van Glabbeek's linear time - branching time spectrum \cite{vanGlabbeekSpectrumI}, one might want to know how weak equivalence compares with trace equivalence, the coarsest equivalence in the spectrum. What can be said is that two HDAs over the same concurrent alphabet will have the same trace language if their underlying automata are trace equivalent. On the other hand, the underlying automata of weakly equivalent HDAs will normally only be trace equivalent up to congruence. Thus, ignoring the higher-dimensional structure of HDAs and comparing only what is comparable, weak equivalence may be considered weaker than trace equivalence. 

As we have pointed out, history-preserving bisimilar HDAs (see \cite{vanGlabbeek}) need not be weakly equivalent. It would be interesting to know under which conditions history-preserving bisimilarity implies weak equivalence.

We have shown that weak equivalence is a congruence with respect to the tensor product and, at least in the coaccessible case, the coproduct of HDAs. This fact and our results on the reduction of HDAs provide means to establish that two HDAs are weakly equivalent. A fundamental problem in this context is whether weak equivalence is decidable for finite HDAs. Given the undecidability of the equivalence problem for regular trace languages \cite{AalbersbergHoogeboom}, it seems likely that weak equivalence is undecidable as well. 

The homology language of an HDA \(\A\) has been defined as a graded submodule of the exterior algebra \(\Lambda(\Sigma_\A)\). An interesting question is which submodules of exterior algebras may actually arise as homology languages of HDAs. It seems possible to show that \(HL(\A)\) is necessarily a graded subcoalgebra of \(\Lambda(\Sigma_\A)\). Assuming that  this is true, the question becomes: Which subcoalgebras of an exterior algebra are homology languages?

The term \emph{weak equivalence} has a particular meaning in homotopy theory. A natural question is thus whether there exists a homotopy theory of HDAs such that two HDAs are weakly equivalent in the sense of this paper if and only if they are weakly equivalent in the homotopy theory.

\newcommand{\etalchar}[1]{$^{#1}$}
\providecommand{\bysame}{\leavevmode\hbox to3em{\hrulefill}\thinspace}
\providecommand{\MR}{\relax\ifhmode\unskip\space\fi MR }
\providecommand{\MRhref}[2]{%
  \href{http://www.ams.org/mathscinet-getitem?mr=#1}{#2}
}
\providecommand{\href}[2]{#2}

\end{document}